\documentclass[10.5pt]{article}

\usepackage[margin=1.1in]{geometry}
\usepackage{lmodern}
\usepackage[T1]{fontenc}
\usepackage[utf8]{inputenc}
\usepackage{amsmath}
\usepackage{amssymb}
\usepackage{amsthm}
\usepackage{booktabs}
\usepackage{graphicx}
\usepackage{fancyvrb}
\usepackage[authoryear,round]{natbib}
\usepackage{bookmark}
\hypersetup{colorlinks=true, linkcolor=blue, citecolor=blue, urlcolor=blue}
\AtBeginDocument{}

\newcommand{\pkg}[1]{\textbf{#1}}
\newcommand{\proglang}[1]{\textsf{#1}}
\makeatletter
\newcommand\code{\bgroup\@makeother\_\@makeother\~\@makeother\$\@codex}
\def\@codex#1{{\normalfont\ttfamily\hyphenchar\font=-1 #1}\egroup}
\makeatother

\newcommand{\fct}[1]{\code{#1()}}

\DefineVerbatimEnvironment{CodeInput}{Verbatim}{fontsize=\small}
\DefineVerbatimEnvironment{CodeOutput}{Verbatim}{fontsize=\small}

\theoremstyle{plain}
\newtheorem{theorem}{Theorem}[section]
\theoremstyle{definition}
\newtheorem{assumption}{Assumption}[section]
\theoremstyle{remark}
\newtheorem*{remark}{Remark}

\title{\pkg{smartcor}: Intelligent Correlation Method Selection for Mixed Variable Types}
\author{M.~Harshvardhan\\American University of Sharjah\\\texttt{harshvardhan@aus.edu} \and Pritam Ranjan\\Indian Institute of Management Indore\\\texttt{pritamr@iimidr.ac.in}}
\date{}

\begin{document}

\maketitle

\begin{abstract}
Pearson correlation is the default measure of association in most statistical software, yet it is only appropriate for pairs of continuous variables with a linear relationship.
When variables are binary, ordinal, or categorical, specialized methods (e.g., point-biserial, polychoric, tetrachoric, and Cram\'{e}r's~$V$) may be more appropriate, but practitioners rarely know which to select.
The \pkg{smartcor} package for \proglang{R} (and its companion \pkg{pysmartcor} package for \proglang{Python}) automatically detects variable types, selects the statistically appropriate correlation method for each pair, and explains its reasoning.
The package supports 14~correlation and association methods covering all 10~variable-type pair combinations, and distinguishes true correlation (for ordinal and continuous pairs) from statistical association (for nominal categorical pairs).
Monte~Carlo simulations validate the selection logic, and a case study with General Social Survey data demonstrates substantive differences between naive and type-aware correlation analysis.

\end{abstract}

\section{Introduction} \label{sec:intro}

Correlation table is among the most fundamental summary statistic computed in applied statistics.
Whether assessing the relationship between income and education, treatment assignment and health outcomes, or customer satisfaction and retention, researchers routinely compute correlation coeﬃcients as a first step in exploratory data analysis.
Yet in practice, the overwhelming majority of applied work defaults to Pearson correlation regardless of variable type.
A large-scale review of more than 18,000 psychology articles found that Pearson correlation was the single most common measure of association, used in more than twice as many studies as the next alternative, Spearman's rank correlation \citep{de2016comparing}.
Pearson's dominance reflects not a considered methodological choice but a software default: \code{cor()} in \proglang{R} \citep{R-base}, \code{numpy.corrcoef()} in \proglang{Python}, and their equivalents in SPSS, Stata, and SAS, all return Pearson correlation by default — even though the method proposed by \citet{pearson1895notes} assumes continuous variables with a linear relationship.

Scalar variables in real datasets come in five fundamental types: continuous, binary, ordinal, count, and categorical. 
A \textit{continuous} variable can take any real value in an interval; temperature, air pressure, and fuel efficiency are typical examples. 
A \textit{binary} variable takes only two values, such as Yes/No or Pass/Fail. An \textit{ordinal} variable carries an inherent ordering but takes only a few discrete labels, and these labels often serve as indicators of an unobserved latent continuous variable; Likert scales are the canonical example. 
A \textit{count} variable takes non-negative integer values \{0, 1, 2, 3, \ldots\} and, unlike ordinal variables, arises from counting rather than ranking, and typically does not have an upper bound. 
Examples include the number of interviews a candidate attends or phone calls received per hour. From a correlation standpoint, all \textit{count} variables are treated as continuous variables and thus not discussed separately at many places in this paper. Finally, a \textit{categorical} variable takes unordered class labels with no inherent ranking, such as country of origin, brand name, or product category. Thus, the total number of variable-type pairs (unordered pairs among continuous, binary, ordinal and categorical) is equal to ${4 \choose 2} +4 = 10$. Correlation or association can be computed for all ten pairwise combinations of these types.
As we detail in \autoref{sec:when}, Pearson correlation formula is correct or equivalent only for three of these ten variable-type combinations.  Pairings among continuous, binary, ordinal, and count variables yield correlation coefficients bounded by $-1\leq r \leq 1$, though the specific estimator differs by pairing.  On the other hand, because categorical variables carry no ordering, their relationship with any other type is measured as \textit{association} on a scale of $0 \leq a \leq 1$.  The available options are detailed in \autoref{tab:correlation-types} and discussed later in this paper.

\begin{table}[!h]
\centering
\caption{Correlation and association measures supported by \pkg{smartcor}, organized by variable-type pair.}
\label{tab:correlation-types}
\begin{tabular}{lcccc}
\toprule
 & Continuous & Binary & Categorical & Ordinal \\
\midrule
Continuous & Pearson & Point-biserial & Cram\'{e}r's $V^{*\dagger}$ & Polyserial \\
 & Spearman &  &  & Spearman \\
 & Kendall's $\tau$ &  &  & Kendall's $\tau$ \\
\midrule
Binary &  & Phi ($\phi$) & Cram\'{e}r's $V^{\dagger}$ & Rank-biserial \\
 &  & Tetrachoric &  &  \\
 &  & Yule's $Q$ &  &  \\
\midrule
Categorical &  &  & Cram\'{e}r's $V^{\dagger}$ & Cram\'{e}r's $V^{\dagger}$ \\
 &  &  & Theil's $U^{\dagger}$ &  \\
 &  &  & Tschuprow's $T^{\dagger}$ &  \\
\midrule
Ordinal &  &  &  & Polychoric \\
 &  &  &  & Kendall's $\tau$ \\
 &  &  &  & Kruskal's $\gamma$ \\
\bottomrule
\multicolumn{5}{l}{\footnotesize $^{*}$Continuous variable binned into quintiles.}\\
\multicolumn{5}{l}{\footnotesize $^{\dagger}$Association measure, not correlation (see text).}
\end{tabular}
\end{table}

A major source of error in social science research is the failure to match the correlation measure to the variables' underlying measurement type \citep{de2016comparing}.
Our mini-study (\autoref{sec:literature}) shows the scale of the problem: Pearson's coefficient accounts for 92.3\% of the correlations reported in economics papers in leading journals.
Yet Pearson is defensible only for the continuous-continuous, continuous-binary, and binary-binary cases; in the latter two it reduces algebraically to the \textit{point-biserial} and \textit{phi} coefficients, as we prove in the \textit{Appendix}.
For every other pairing it returns a number without a meaningful interpretation.
The recurring version of this error is treating an ordinal or count variable as continuous, since discreteness is easily mistaken for continuity, and doing so can distort estimates and support misleading conclusions.
Using Pearson correlation for ordinal variables actually is the most common error in social science research, as evidenced from our mini-study in \autoref{sec:literature}.
In reality, what practitioners often lack is clear guidance on pairing the measure to the data, which \texttt{smartcor} supplies automatically.
Thus, using the wrong method leads to three categories of errors:

\begin{enumerate}
  \item \textbf{Attenuation.} When continuous constructs are observed through ordinal proxies such as Likert scales, Pearson correlation systematically underestimates the true association. Our simulations (\autoref{sec:attenuation}) show that with five ordinal categories and true $\rho = 0.6$, Pearson correlation underestimates $\rho$ by approximately $0.05$ (a relative bias of about $-9\%$), while polychoric correlation recovers the true value with near-zero bias.
  \item \textbf{Meaningless results.} Applying Pearson correlation to nominal variables (religion, region, brand name) produces numbers without meaningful interpretation, since the numeric codes are arbitrary. For such pairs, association measures like Cram\'{e}r's~$V$ \citep{cramer1946} can be used, and standard correlation measures do not work.
  \item \textbf{Untested assumptions.} Polychoric and tetrachoric correlations recover latent associations between ordinal pairs but assume underlying bivariate normality \citep{poon1987maximum, drasgow1986polychoric}. When this assumption is violated (i.e., when the latent distribution is skewed or heavy-tailed), these estimators exhibit positive bias. Practitioners rarely test or acknowledge this assumption.
\end{enumerate}

These errors are compounded by a fragmented software landscape.
Computing appropriate correlations for a mixed-type dataset in \proglang{R} currently requires combining functions from multiple packages: \code{cor()} from \pkg{stats} for Pearson, Spearman, and Kendall; \code{polychor()} and \code{polyserial()} from \pkg{polycor} \citep{R-polycor} for latent-variable methods; \code{tetrachoric()} from \pkg{psych} \citep{R-psych} for binary pairs; and manual contingency-table calculations for Cram\'{e}r's~$V$, Theil's~$U$, and related association measures, as they didn't have an immediate built-in function in \proglang{R}/\proglang{Python}.
Closest to our work is the \pkg{correlation} package \citep{R-correlation} which provides a unified interface but does not auto-detect all four variable types, does not cover all 10~type-pair combinations, and does not explain its method choices.
As we stated earlier, we distinguish \emph{correlation} (when both variables are \textit{continuous}, \textit{binary}, \textit{ordinal} or \textit{count}) from \emph{association} (when at least one variable is \textit{categorical}) between a pair of variables; existing software rarely makes this distinction explicit.

The \pkg{smartcor} package for \proglang{R} and \proglang{Python} addresses this gap.\footnote{Unless stated otherwise, every method in \pkg{smartcor} for \proglang{R} has an equivalent in \pkg{pysmartcor} for \proglang{Python}.} Given any pair of variables or a full data frame of mixed types, \pkg{smartcor} \textbf{detects} each variable's type (continuous, binary, ordinal, count, or categorical) using a configurable heuristic; \textbf{selects} the statistically appropriate method from a library of 14~correlation and association measures, with selection logic grounded in the simulation evidence presented in \autoref{sec:when}; \textbf{tests} for latent normality assumption between two ordinal variables to pick most suitable method, and \textbf{explains} its reasoning through a human-readable \code{rationale} field.
Users retain full control: any automatic choice can be \textbf{overridden} by specifying a method explicitly.
A companion \proglang{Python} package with an identical API is described in \autoref{sec:python}.

The remainder of this paper is organized as follows.
\autoref{sec:when} presents the statistical background: when method choice matters, when Pearson correlation suffices, and how attenuation and violated normality assumptions affect estimates.
\autoref{sec:package} describes the \pkg{smartcor} \proglang{R} package, including its design, type detection, method selection, and core functions.
\autoref{sec:casestudy} demonstrates the package on General Social Survey data.
\autoref{sec:comparison} compares \pkg{smartcor} with existing software.
\autoref{sec:python} describes the \proglang{Python} companion package.
\autoref{sec:summary} concludes with limitations and future directions.

\section{When does the choice of correlation method matter?} \label{sec:when}

Correlation measures the strength and direction of statistical dependence between two random variables.
Its value typically ranges between $-1$ and $+1$, with $0$ indicating no association.
The appropriate measure depends on the types of variables being correlated.
Previously presented \autoref{tab:correlation-types} summarizes the 14 different correlation methods supported by \pkg{smartcor}, organized by variable-type pair.

\noindent
An important distinction underlines \autoref{tab:correlation-types}: the difference between \emph{correlation} and \emph{association}.
\emph{Correlation} quantifies the strength and direction of a monotonic or linear relationship between two variables.
It requires that both variables possess at least an ordinal measurement scale---that is, their values must have a meaningful ordering (continuous, count, binary, or ordinal).
A Pearson coefficient of $r = -0.7$ tells us that higher values of one variable have strong chance to pair-up with lower values of the other, and the sign carries directional information.
All correlation methods in Table~\ref{tab:correlation-types} without a dagger ($\dagger$) symbol produce signed values in $[-1, +1]$.

\emph{Association}, by contrast, measures statistical dependence between variables whose levels have no inherent ordering.
When a variable is categorical (nominal)---religion, region, treatment group---the numeric codes are arbitrary, and neither direction nor monotonicity is defined.
Some binary variables may be categorical: when a scalar variable with two (and only two) unique values exists, it is a binary categorical variable and must use an association measure instead of correlation measures.
Measures such as Cram\'{e}r's~$V$, Theil's~$U$, and Tschuprow's~$T$ test a different question: are the two variables independent, or does knowing the level of one reduce uncertainty about the other?
These measures are bounded between $0$ (complete independence) and $1$ (perfect association), have no sign, and cannot be interpreted as directional relationships.
A Cram\'{e}r's~$V$ of $0.3$ means the variables are moderately dependent, but it says nothing about which levels covary or in which direction.
All methods in Table~\ref{tab:correlation-types} with a dagger ($\dagger$) symbol produce unsigned values on $[0, +1]$.

Three practical consequences follow from this distinction:
\begin{enumerate}
  \item \textbf{No direct comparison.} A polychoric correlation of $0.5$ and
    a Cram\'{e}r's~$V$ of $0.5$ do not mean the same thing. The former
    estimates a latent linear relationship; the latter quantifies departure
    from independence. Combining both in a single matrix (as
    \code{smart\_cormat()} does) is useful for exploratory purposes, but the
    values should not be ranked against each other as if they were on a
    common scale.
  \item \textbf{When to use association measures.} Cram\'{e}r's~$V$ and its
    relatives are appropriate when at least one variable is truly nominal:
    its categories cannot be meaningfully ordered (e.g., country, species,
    diagnosis code). If the variable has a natural ordering that the analyst
    has chosen to ignore (e.g., education coded as a factor rather than an
    ordered factor), recoding it as ordinal and using a correlation method
    is preferable.
  \item \textbf{When to ignore them.} If the research question concerns
    the strength of a directional relationship (``does higher education
    predict higher income?''), association measures cannot answer it.
    They detect dependence but not direction. In such cases, computing
    Cram\'{e}r's~$V$ between an ordinal and a categorical variable
    discards the ordinal information and should be treated as a rough
    screening tool rather than a definitive result.
\end{enumerate}

The most commonly used correlation methods fall into three families.
\emph{Pearson correlation} \citep{pearson1895notes, benesty2009pearson} measures linear relationship between two continuous variables; it is a parametric method assuming bivariate normality.
\emph{Spearman's rank correlation} \citep{spearman1904proof} and \emph{Kendall's~$\tau$} \citep{kendall1938new, kendall1945treatment} are nonparametric rank-based methods that capture monotonic relationships; they are robust to outliers and applicable to ordinal data.
\emph{Latent-variable methods} (polychoric \citep{poon1987maximum, drasgow1986polychoric}, polyserial, and tetrachoric correlation) assume that observed ordinal or binary variables arise from underlying continuous normal latent variables and estimate the correlation between those latent variables.
Finally, \emph{categorical association measures} including Cram\'{e}r's~$V$ \citep{cramer1946}, Theil's~$U$ \citep{theil1967economics}, and Tschuprow's~$T$ \citep{tschuprow1939principles} quantify dependence for unordered categorical variables through chi-square or entropy-based statistics; they belong to a fundamentally different class than the correlation methods discussed above.

When is it safe to simply use Pearson?
When does the choice matter?
The following subsections answer these questions: first by defining each method precisely, then by presenting simulation evidence, with mathematical proofs deferred to Appendix~\ref{app:pearson-biserial} and Appendix~\ref{app:pearson-phi}.

Among all 10 variable-type pairs discussed in \autoref{tab:correlation-types} Pearson correlation is the correct choice for exactly three of them: continuous--continuous (the textbook case), continuous--binary (where Pearson coincides with the point-biserial coefficient under 0/1 coding), and binary--binary (where Pearson coincides with the phi ($\phi$) coefficient). Theorem~2.1 and Theorem~2.2 in \autoref{sec:equiv} establish the equivalence of Pearson with point-biserial and phi, respectively. For the remaining seven combinations, Pearson is biased, attenuated, or undefined: it underestimates the association for binary and ordinal variables (\autoref{sec:attenuation}), it ignores the latent structure of ordinal scales, and it is mathematically undefined when one of the variables is unordered categorical.
A practitioner who reaches for \code{cor()} in \proglang{R}, \code{numpy.corrcoef()} in \proglang{Python}, or similar functions in Stata or SAS, on a heterogeneous data set, therefore, has at most a $3/10$ chance of computing the appropriate measure for each pair and an overwhelming chance of reporting an attenuated or biased value somewhere in the matrix.
This is the central motivation for \pkg{smartcor}: routing each pair of variables to its correct estimator should not require the practitioner to remember which of the seven other cells they are in.

\subsection{Correlation methods by variable type} \label{sec:methods}

Each of the 14~methods in Table~\ref{tab:correlation-types} rests on distinct mathematical foundations.
This subsection provides concise definitions organized by the types of variables being correlated.
Readers who are already familiar with these formulas may skip ahead to \autoref{sec:equiv}.

\paragraph{(a) Continuous--continuous pairs.}

The \emph{Pearson correlation coefficient}, introduced by \citet{pearson1895notes}, measures the strength of the linear relationship between two continuous paired variables.
For a paired random variable ($X$, $Y$) with mean ($\mu_X, \mu_Y$) and standard deviation ($\sigma_X, \sigma_Y$), the population correlation coefficient is defined as
\begin{equation} \label{eq:pearson}
  \rho = \frac{\operatorname{Cov}(X, Y)}{\sigma_X \, \sigma_Y}
       = \frac{E[(X - \mu_X)(Y - \mu_Y)]}{\sigma_X \, \sigma_Y},
\end{equation}
with the sample version given by
\begin{equation} \label{eq:pearson_sample}
  r = \frac{\sum_{i=1}^{n} (x_i - \bar{x})(y_i - \bar{y})}
           {\sqrt{\sum_{i=1}^{n} (x_i - \bar{x})^2} \, \sqrt{\sum_{i=1}^{n} (y_i - \bar{y})^2}},
\end{equation}
where $\{(x_i, y_i), i=1,...,n\}$ is the observed paired sample and $\bar{x}, \bar{y}$ are the sample means. The denominator is the product of the sample standard deviations (up to the shared $n-1$ or $n$ factor, which cancels), so this expression is the direct sample analogue of \autoref{eq:pearson}.
The coefficient value ranges from $-1$ (perfect negative linear relationship) to $+1$ (perfect positive linear relationship), with $0$ indicating no linear dependence; the bound $|\rho| \leq 1$ follows from the Cauchy-Schwarz inequality
\citep{casella2002statistical}.
Pearson correlation is a parametric measure: it is most appropriate when the bivariate distribution is approximately normal and the relationship is linear \citep{benesty2009pearson}.

\emph{Spearman's rank correlation} \citep{spearman1904proof} is a nonparametric alternative that captures monotonic (not necessarily linear) relationship.
It is computed by applying the Pearson formula to the ranks of the observations rather than to their raw values.
When there are no tied ranks, a convenient shortcut is
\begin{equation} \label{eq:spearman}
  r_s = 1 - \frac{6 \sum_{i=1}^{n} d_i^2}{n(n^2 - 1)},
\end{equation}
where $d_i$ is the difference between the ranks of the $i$-th paired observations. Since it operates on ranks, Spearman correlation is robust to outliers and does not require normality.

\emph{Kendall's~$\tau$} \citep{kendall1938new} takes a different approach, evaluating monotonic relationship through pairwise comparisons.
A pair of observations $(x_i, y_i)$ and $(x_j, y_j)$ is said to be \emph{concordant} if $(x_i - x_j)(y_i - y_j) > 0$ and \emph{discordant} if the product is negative.
The tie-corrected variant $\tau_b$ \citep{kendall1945treatment}, which is the default in most software, is
\begin{equation} \label{eq:kendall}
  \tau_b = \frac{N_c - N_d}
               {\sqrt{(N_c + N_d + T_X)(N_c + N_d + T_Y)}},
\end{equation}
where $N_c$ and $N_d$ are the numbers of concordant and discordant pairs, and $T_X, T_Y$ count tied pairs on each variable.
Since $\tau_b$ has a direct probabilistic interpretation (the difference between the probability of concordance and discordance), it is often preferred for ordinal data and small samples.

\paragraph{(b) Continuous--binary pairs.}

The \emph{point-biserial correlation} measures relationship between a continuous variable $X$ and a paired binary variable $Y \in \{0, 1\}$.
It is defined as
\begin{equation} \label{eq:pointbiserial}
  r_{pb} = \frac{\bar{x}_1 - \bar{x}_0}{s_x}
           \sqrt{\frac{n_1 \, n_0}{n^2}},
\end{equation}
where $\bar{x}_1$ and $\bar{x}_0$ are the means of $X$ in the two groups, $s_x$ is the overall standard deviation, and $n_1, n_0$ are the group sizes with $n = n_1 + n_0$.
The formula is algebraically identical to applying Pearson correlation when $Y$ is coded 0/1 (see Appendix~\ref{app:pearson-biserial} for the proof).
In practice, no separate computation is needed; Pearson correlation on the raw data yields the same value.

\paragraph{(c) Binary--binary pairs.}
The \emph{phi coefficient} ($\phi$) \citep{pearson1904contingency} quantifies association between two binary variables. Let $N_{ij}$ denote the count of observations with $X = i$ and $Y = j$ in the $2 \times 2$ contingency table, for $i, j \in \{0, 1\}$, then
\begin{equation} \label{eq:phi}
  \phi = \frac{N_{11} N_{00} - N_{10} N_{01}}
              {\sqrt{(N_{11} + N_{10})(N_{01} + N_{00})(N_{11} + N_{01})(N_{10} + N_{00})}}.
\end{equation}
The numerator contrasts concordant cells ($N_{11}, N_{00}$) with discordant cells ($N_{10}, N_{01}$), and the denominator is the geometric mean of the four marginal totals.
Like point-biserial correlation, $\phi$ is algebraically equivalent to Pearson correlation when both variables are coded $0/1$.
It ranges from $-1$ to $+1$, but attaining either extreme requires equal marginal probabilities: $\phi = +1$ only when $P(X=1) = P(Y=1)$, and $\phi = -1$ only when $P(X=1) = P(Y=0)$.
When the marginal probabilities are unequal, the maximum attainable $|\phi|$ is strictly less than one and is given by
\begin{equation} \label{eq:phimax}
  |\phi|_{\max} = \sqrt{\frac{\min(p_X, p_Y)\,\min(1 - p_X,\, 1 - p_Y)}
                             {\max(p_X, p_Y)\,\max(1 - p_X,\, 1 - p_Y)}},
\end{equation}
where $p_X = P(X = 1)$ and $p_Y = P(Y = 1)$ \citep{davenport1991phi}.
This dependence on equality of marginal probabilities is the central reason why \textit{Yule's~$Q$} or \textit{tetrachoric} correlation is often preferred when the binary variables are dichotomized continuous quantities.

\emph{Yule's~$Q$} \citep{yule1912methods} is an odds-ratio-based alternative that is invariant to marginal distributions:
\begin{equation} \label{eq:yuleq}
  Q = \frac{N_{11} N_{00} - N_{10} N_{01}}{N_{11} N_{00} + N_{10} N_{01}} = \frac{OR - 1}{OR + 1},
\end{equation}
where $OR = N_{11} N_{00} / (N_{10} N_{01})$ is the odds ratio.
Unlike $\phi$, Yule's~$Q$ can reach $\pm 1$ regardless of whether $p_X$ and $p_Y$ are equal, making it preferable when the marginal probabilities differ substantially across the two variables.

\emph{Tetrachoric correlation} is a latent-variable method which assumes that the two binary outcomes arise from discretizing an underlying bivariate normal distribution at unknown thresholds \citep{hershberger2005tetrachoric}.
It estimates the Pearson correlation between the latent continuous variables via maximum likelihood.
Although no closed-form expression exists, the method uses a two-step procedure that first estimates thresholds from the marginal proportions, and then maximizes the bivariate normal likelihood over $\rho$.
It is the binary special case of polychoric correlation and is discussed in Part (e) of this section \citep{kiwanuka2022polychoric}.

\paragraph{(d) Categorical--categorical pairs.}

\emph{Cram\'{e}r's~$V$} \citep{cramer1946} normalizes the Pearson chi-square statistic to produce a measure of association for a pair of nominal variables with $r$ and $c$ levels respectively.
For a contingency table with $r$~rows and $c$~columns, the Pearson chi-square statistic is

\begin{equation} \label{eq:pearsonchisq}
  \chi^2 = \sum_{i=1}^{r} \sum_{j=1}^{c}
  \frac{(O_{ij} - E_{ij})^2}{E_{ij}},
\end{equation}

where $O_{ij}$ is the observed count in cell $(i,j)$. Under the null hypothesis that the two variables are independent, the corresponding expected count is
$E_{ij}=\frac{n_{i\cdot}n_{\cdot j}}{n}$,
where $n_{i\cdot}=\sum_{j=1}^{c}O_{ij}$ and $n_{\cdot j}=\sum_{i=1}^{r}O_{ij}$ are the row and column marginal totals, respectively.
Cram\'{e}r's~$V$ is then defined as

\begin{equation} \label{eq:cramersv}
  V = \sqrt{\frac{\chi^2}{n \cdot \min(r-1, \, c-1)}},
\end{equation}

where $n$ is the total sample size.
The coefficient ranges from $0$ (independence) to $1$ (perfect association) and is always nonnegative because nominal categories have no inherent direction.

\emph{Tschuprow's~$T$} \citep{tschuprow1939principles} is a chi-square-based measure that uses the geometric mean of the degrees of freedom as its normalizer:
\begin{equation} \label{eq:tschuprowt}
  T = \sqrt{\frac{\chi^2}{n \sqrt{(r-1)(c-1)}}}.
\end{equation}
It ranges from $0$ to $1$ but can attain unity only when the contingency table is square ($r = c$).
For nonsquare tables, the maximum value of $T$ is strictly less than one, complicating interpretation.
In modern practice, Cram\'{e}r's~$V$ is generally preferred because of this reason.

\emph{Theil's~$U$} \citep{theil1967economics}, also called the \textit{uncertainty coefficient}, applies information theory to measure how much knowing one variable reduces uncertainty about the other:
\begin{equation} \label{eq:theilsu}
  U(Y|X) = \frac{H(Y) - H(Y|X)}{H(Y)},
\end{equation}
where $H(Y) = -\sum_j p(y_j) \ln p(y_j)$ is the entropy of $Y$ and $H(Y|X)$ is the conditional entropy.
Theil's~$U$ ranges from $0$ to $1$.
Crucially, it is asymmetric: $U(Y|X)$ need not equal $U(X|Y)$.
This makes it suitable when one variable is conceptually the predictor and the other the outcome.
An analytical asymptotic standard error is available but is numerically unstable: it contains a $\log(p_{ab})$ term, where $p_{ab}=P(X=x_a, Y=y_b)$ is the joint cell probability, that diverges in sparse cells. 
\code{DescTools::UncertCoef} mitigates this with a small-cell correction (\texttt{p.zero.correction} $= 1/n^2$); \pkg{smartcor} sidesteps the problem entirely using a percentile bootstrap as follows. 
The default confidence interval is a percentile bootstrap on $U$, and the default $p$-value is obtained from a permutation test under the null hypothesis that $X$ and $Y$ are independent ($B = 500$ shuffles of $y$; recompute $U$ and count the proportion at least as large as the observed value).
Both treat $U$ as a non-negative measure for which a one-sided test against the independence null is the natural inference.

\paragraph{(e) Ordinal--ordinal pairs.}

\emph{Polychoric correlation} \citep{drasgow1986polychoric, poon1987maximum} estimates the Pearson correlation between two latent continuous variables that are assumed to be bivariate normal but observed only through ordinal categories.
The observed ordinal variable with $k$ categories is modeled as arising from $k-1$ thresholds $\tau_1, \ldots, \tau_{k-1}$ applied to the latent variable.
Both the thresholds and the latent correlation $\rho$ are estimated jointly via maximum likelihood.
No closed-form solution exists.
The key benefit is that polychoric correlation recovers the ``true'' latent association that Pearson correlation on ordinal data would systematically underestimate, with attenuation worsening as the number of categories decreases.

\emph{Goodman--Kruskal gamma} ($\gamma$) \citep{goodman1954measures} is a nonparametric measure for ordinal variables based on concordant and discordant pairs:
\begin{equation} \label{eq:gamma}
  \gamma = \frac{N_c - N_d}{N_c + N_d}.
\end{equation}
Unlike Kendall's~$\tau_b$, gamma excludes tied pairs from the denominator entirely.
This means it can overstate the strength of association when ties are common: $\gamma$ estimates $P(\text{concordant}) - P(\text{discordant})$ conditional on the pair being untied, not the unconditional difference that $\tau_b$ targets.
Kendall's~$\tau_b$ (Equation~\ref{eq:kendall}) is also widely used for ordinal--ordinal pairs, particularly when ties are prevalent and a more conservative estimate is desired.

\paragraph{(f) Mixed pairs: continuous--ordinal and binary--ordinal.}

\emph{Polyserial correlation} \citep{olsson1982polyserial} extends the latent-variable framework to pairs where one variable is continuous and the other is ordinal.
It assumes the ordinal variable arises from thresholding a latent normal variable and estimates the correlation between the observed continuous variable and the latent variable via maximum likelihood.
Like polychoric correlation, it corrects for the attenuation introduced by discretization and requires the latent normality assumption.

The \emph{rank-biserial correlation} \citep{cureton1956rank} quantifies association between a binary variable and an ordinal variable by comparing the mean ranks of the ordinal variable across the two binary groups.
It is a nonparametric method closely related to the Mann--Whitney~$U$ statistic, and it requires no distributional assumptions.

For continuous--ordinal pairs, Spearman's $r_s$ and Kendall's $\tau_b$ also apply directly by ranking the continuous variable and proceeding as in the ordinal--ordinal case.

\subsection{Equivalences: when Pearson suffices} \label{sec:equiv}

Two important equivalences hold when binary variables are coded as 0/1.

\begin{assumption}
Let $\{(x_i, y_i), i=1,...,n\}$ be a sample of $n$ pairs of observations satisfying:
(i) $n \geq 2$;
(ii) $y_i \in \{0,1\}$ with both groups non-empty, so $n_1 = \sum_i y_i$ and $n_0 = n - n_1$ both lie in $\{1, \ldots, n-1\}$, equivalently $\hat{p} = \hat{P}(Y=1)= n_1/n \in (0,1)$;
(iii) $x_i \in \mathbb{R}$ has positive sample variance, $s_x^2 > 0$.
\end{assumption}

\begin{theorem}[Pearson equals point-biserial]
Under Assumption~2.1, the Pearson sample correlation between $x$ and $y$ coincides with the point-biserial coefficient,
\[
r \;=\; \frac{\bar{x}_1 - \bar{x}_0}{s_x}\sqrt{\hat{p}(1-\hat{p})} \;=\; r_{pb},
\]
where $\bar{x}_j = n_j^{-1}\sum_{y_i=j} x_i$ for $j \in \{0,1\}$. 
\end{theorem}

The proof is in Appendix~\ref{app:pearson-biserial}.

\begin{assumption}
Let $\{(x_i, y_i), i=1,...,n\}$ be a sample of $n$ pairs of observations satisfying:
(i) $n \geq 2$;
(ii) $x_i, y_i \in \{0,1\}$;
(iii) both marginals are non-degenerate, i.e., neither row sum nor column sum of the $2\times 2$ contingency table equals $0$ or $n$.
\end{assumption}

\begin{theorem}[Pearson equals $\phi$]
Under Assumption~2.2, the Pearson sample correlation between $x$ and $y$ equals the $\phi$ coefficient,
\[
r \;=\; \frac{N_{11}N_{00} - N_{10}N_{01}}{\sqrt{N_{1\cdot}\,N_{0\cdot}\,N_{\cdot 1}\,N_{\cdot 0}}} \;=\; \phi,
\]
where $N_{ij}$ is the count of observations with $x = i$ and $y = j$, and $N_{i\cdot}$, $N_{\cdot j}$ denote row and column marginals respectively.
\end{theorem}

The proof is in Appendix~\ref{app:pearson-phi}.

Both equivalences are algebraic identities: the formulas reduce to one another when the binary variable is coded $0/1$, with no statistical assumptions invoked.
Distributional assumptions (bivariate normality, linearity) are required to interpret the resulting coefficient as a population parameter and to construct valid confidence intervals, but they do not affect the equivalence itself.
The practical implication is that no separate computation is needed for these special cases: calling \code{cor} on raw $0/1$-coded variables returns the point-biserial coefficient when paired with a continuous variable and the $\phi$ coefficient when paired with another binary variable.
The equivalence does not extend further. 
For ordinal-ordinal, ordinal-continuous, or nominal pairs, Pearson on the raw codes is no longer algebraically identical to the type-appropriate coefficient (polychoric, polyserial, Cram\'er's $V$, etc.), and the choice of method materially changes the reported value.

\subsection{Attenuation: when Pearson underestimates} \label{sec:attenuation}

When binary or ordinal variables represent discretized versions of underlying continuous variables, Pearson correlation systematically underestimates the true latent association.
This phenomenon, known as \emph{attenuation}, becomes more severe as the number of response categories decreases. This section presents Monte Carlo simulation-based comparison for several variable-type pairs. 

\paragraph{Binary-binary pairs.}
\autoref{fig:phi-tetrachoric} compares the phi coefficient (which equals Pearson for 0/1 data) with tetrachoric correlation for recovering the latent association between two dichotomized normal variables. Sweeping the true correlation across $\rho \in \{0.1, 0.2, \ldots, 0.9\}$ at $n = 500$ with $500$ replications per cell, the picture is unambiguous: the phi coefficient underestimates the latent correlation by an amount that grows with $|\rho|$---at $\rho = 0.8$ the mean estimate is only $\hat\phi \approx 0.59$ for balanced binary marginals, a downward bias of about $0.21$. Tetrachoric correlation centers on the true value with near-zero bias when the latent normality assumption holds. Thus, Tetrachoric outperforms the Phi (and Pearson) correlation coefficient

\begin{figure}[!h]
\centering
\includegraphics[width=0.85\textwidth]{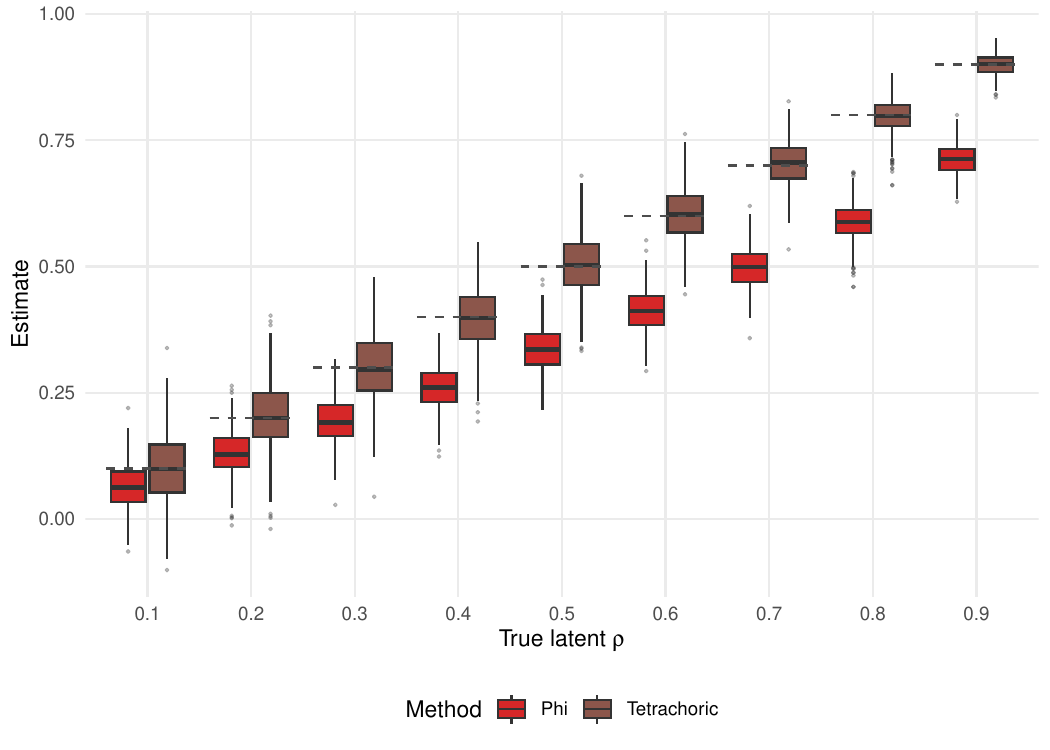}
\caption{Boxplots of estimated correlation across $500$ replications, by method (phi versus tetrachoric) and true latent $\rho$. Horizontal dashed line marks the true value at each panel; phi consistently sits below it, tetrachoric centers on it.}
\label{fig:phi-tetrachoric}
\end{figure}

\paragraph{Ordinal-ordinal pairs.}
\autoref{fig:ordinal-bias-normal} shows the bias (estimate minus true $\rho$) for four correlation methods applied to ordinal variables with 5~categories and an underlying bivariate normal latent distribution, at $n = 500$ across $500$ replications.

\begin{figure}[!h]
\centering
\includegraphics[width=0.95\textwidth]{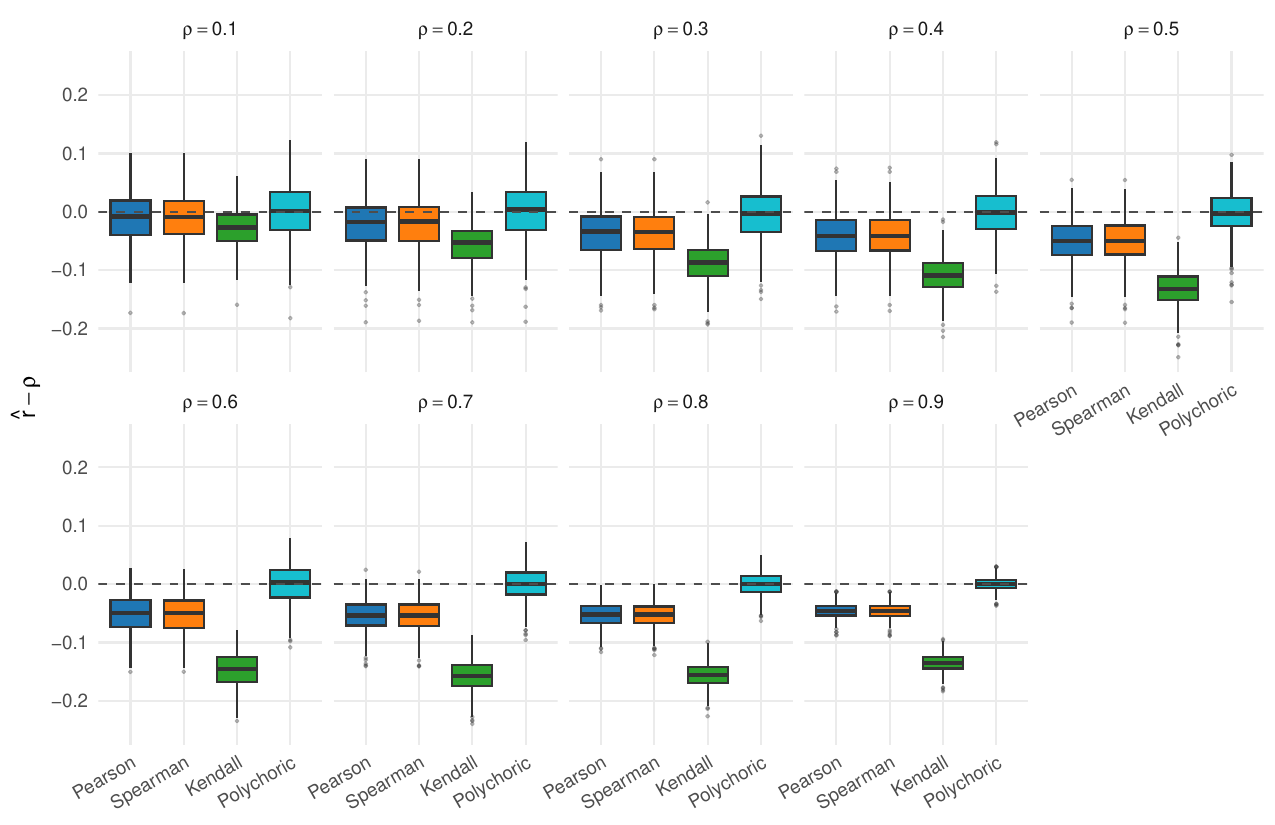}
\caption{Bias of four estimators on ordinal--ordinal data ($n = 500$, $k = 5$ Likert categories, normal latent), faceted by true $\rho$. Polychoric correlation centers on zero across the full range of $\rho$; Pearson and Spearman exhibit moderate attenuation that grows with $|\rho|$; Kendall's~$\tau$ sits well below zero, reflecting its different scale rather than distributional sensitivity.}
\label{fig:ordinal-bias-normal}
\end{figure}

Polychoric correlation is approximately unbiased across all correlation strengths; its boxes hug the zero reference line at every $\rho$.
Pearson and Spearman show moderate attenuation bias of about $-0.05$ for moderate to strong correlations, and the bias is systematic rather than noisy: the entire distribution shifts downward, not just the tails.
Kendall's~$\tau$ exhibits the largest absolute deviation due to its different scale ($\tau \neq \rho$ even in the population), so the comparison should be read as ``what does each estimator return when applied to ordinal data?'' rather than ``which estimator is closest to the true value?''.

\paragraph{Effect of number of ordinal categories.}
\autoref{fig:categories-effect} isolates how the number of Likert response options affects the polychoric estimator. Sweeping the common scale lengths $K \in \{3, 5, 7, 9\}$ at $n = 500$, polychoric correlation is unbiased at every length (its boxes hug the zero line throughout) but its \emph{precision} improves steadily as categories are added: at $\rho = 0.8$ the standard deviation of the estimate falls from $0.025$ at $K = 3$ to $0.021$ at $K = 5$ and $0.018$ at $K = 9$. The gains flatten as the scale lengthens: moving from a 3-point to a 5-point scale buys about as much precision as all further additions combined.

\begin{figure}[h]
\centering
\includegraphics[width=0.85\textwidth]{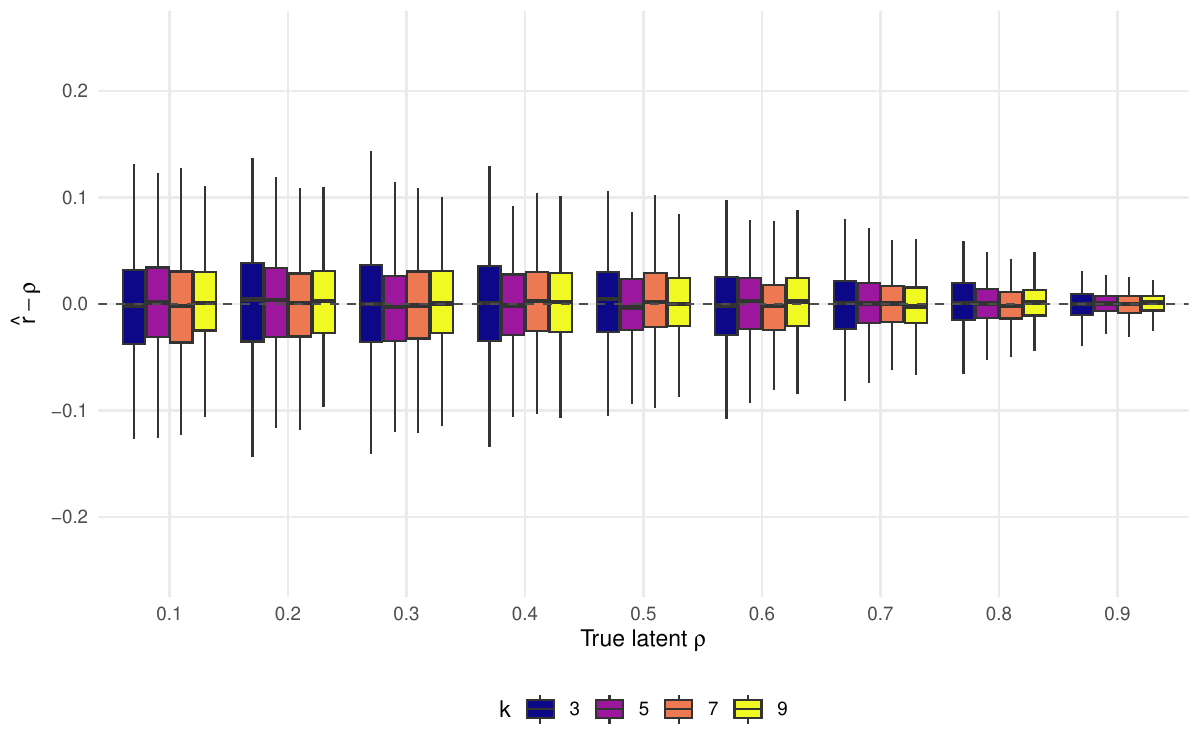}
\caption{Bias of the polychoric estimator grouped by true latent $\rho$ ($n = 500$), colored by Likert scale length $k$. Polychoric is unbiased at every $\rho$ regardless of scale length; what improves with more categories is precision, i.e. the spread of the estimate shrinks as $k$ grows, with most of the gain achieved by a 5-point scale.}
\label{fig:categories-effect}
\end{figure}

For practitioners designing instruments, the figure delivers a concrete recommendation: a 5-point Likert scale already captures much of the recoverable precision, and adding further options yields diminishing returns that may not justify the extra cognitive load on respondents. Since polychoric correlation is unbiased at every $K$, the choice of scale length matters most for the \emph{variance} of the polychoric estimate and, for analysts who fall back on Pearson or Spearman, for the magnitude of attenuation, which in the same simulation shrinks from about $-0.11$ at $K = 3$ to $-0.03$ at $K = 9$ when $\rho = 0.8$.

\paragraph{Continuous-ordinal pairs.}
\autoref{fig:cont-ordinal} shows bias for continuous--ordinal pairs, where polyserial correlation recovers the true correlation under latent normality. The pattern parallels the ordinal--ordinal case: Pearson and Spearman attenuate by similar amounts (e.g., at $\rho = 0.8$ they give $\hat\rho \approx 0.75$ and $0.76$ respectively, a downward bias of $\approx 0.04$), and polyserial centers on the true value across the full range of $\rho$.

\begin{figure}[!h]
\centering
\includegraphics[width=0.95\textwidth]{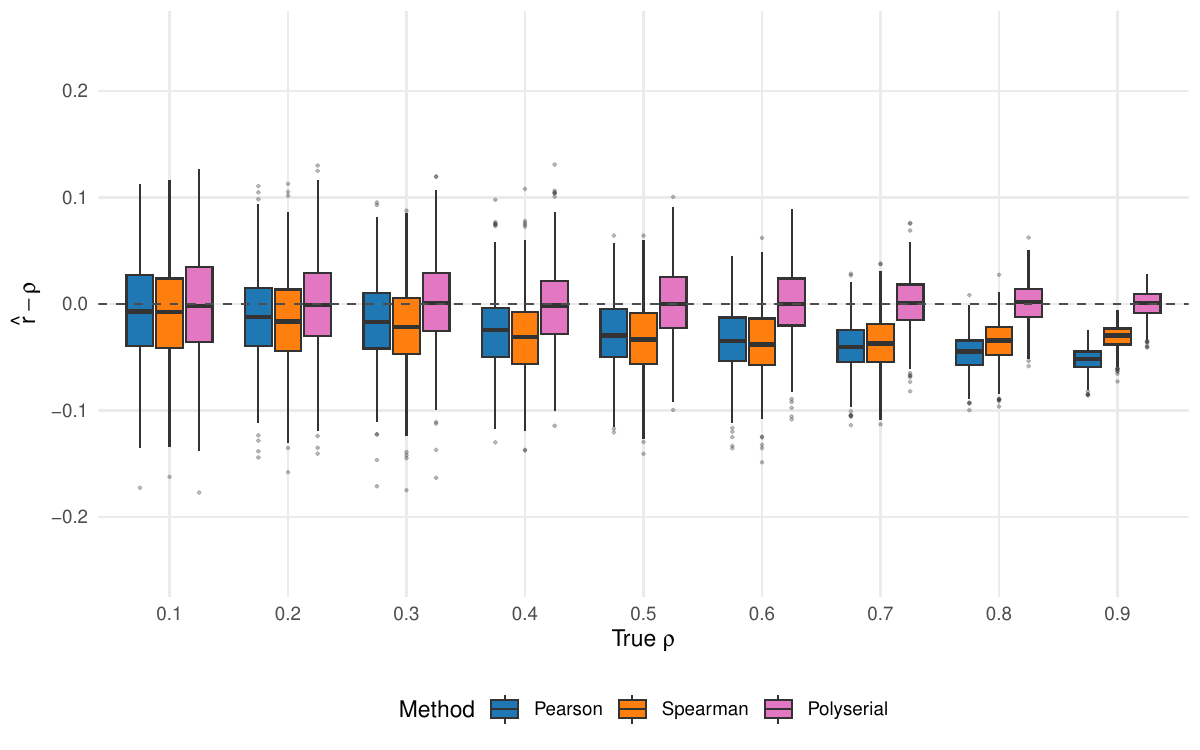}
\caption{Bias of three estimators on continuous--ordinal data ($n = 500$, $k = 5$, normal latent), grouped by true $\rho$. Polyserial recovers the latent correlation under bivariate normality; Pearson and Spearman attenuate by roughly $4$--$8\%$ of the true value.}
\label{fig:cont-ordinal}
\end{figure}

\subsection{The latent normality question} \label{sec:normality}

Polychoric correlation (for ordinal-ordinal pairs) and tetrachoric correlation (for binary-binary pairs) assume that the observed variables arise from discretizing bivariate normal latent variables.
When this assumption holds, these estimators are unbiased (as shown in \autoref{fig:phi-tetrachoric} and \autoref{fig:ordinal-bias-normal}). But what happens when the latent distribution is continuous and unimodal but non-normal?

\autoref{fig:polychoric-nonnormal} shows the bias of four methods under six different latent distributions, with the ordinal variables having five categories and the true latent correlation swept across $\rho \in \{0.1, 0.2, \ldots, 0.9\}$.
Crucially, we separate \emph{right-skewed} from \emph{left-skewed} latents (rather than lumping ``skewed'' as a single category), and we separate \emph{symmetric heavy-tailed} ($t_3$-like) from \emph{asymmetric heavy-tailed} (a $70/30$ mixture of two normals with unequal variances and a shift). 
These distinctions matter because polychoric bias is directional, not symmetric.

\begin{figure}[!h]
\centering
\includegraphics[width=\textwidth]{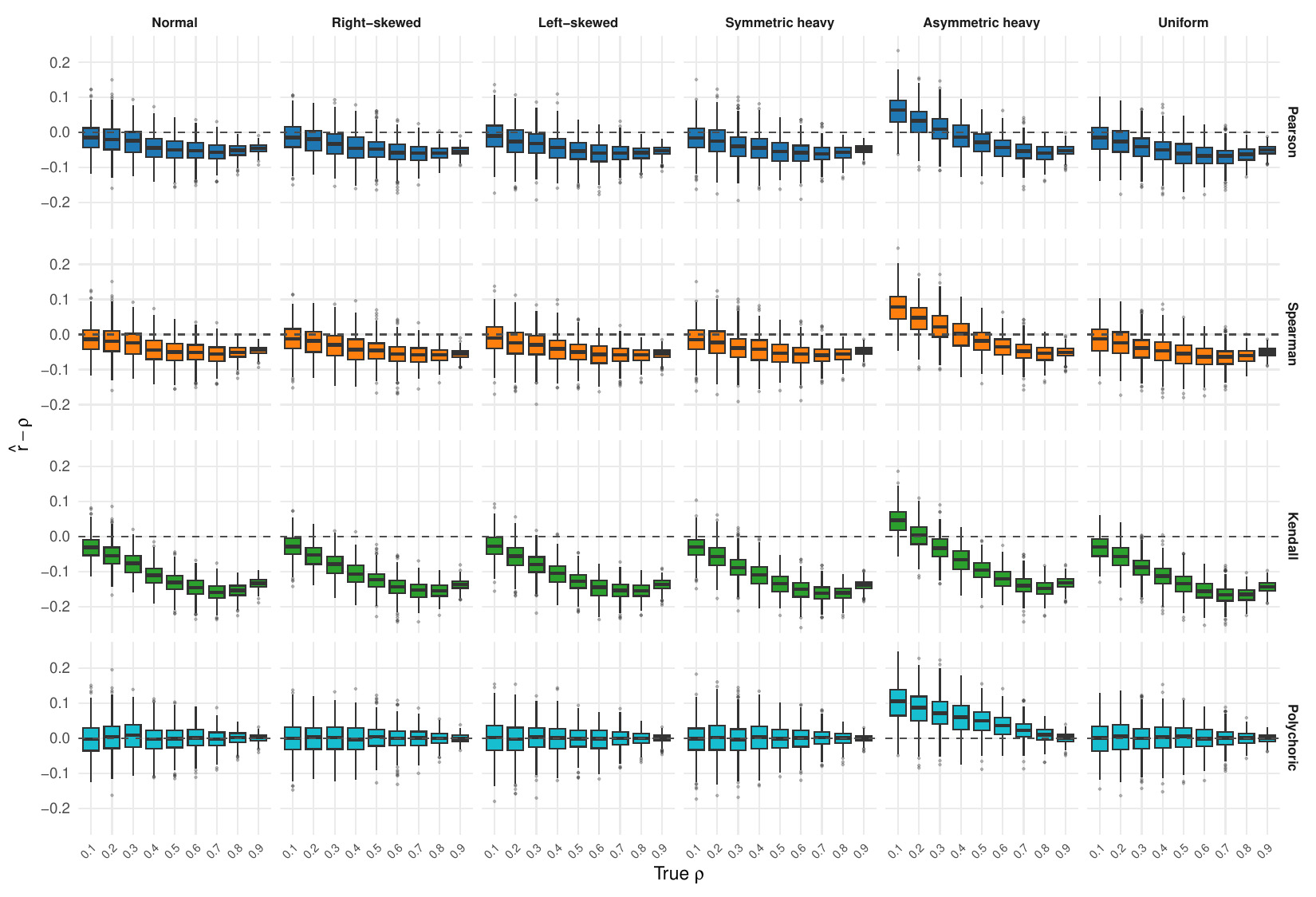}
\caption{Bias by latent distribution type ($n = 500$, $k = 5$, $\rho \in \{0.1, 0.2, \ldots, 0.9\}$, $500$ replications). Rows: methods (Pearson, Spearman, Kendall, polychoric). Columns: latent distribution; all panels share a common $y$-axis clipped at $\pm 0.25$ for readability. Skewed distributions are split into right- and left-skewed; heavy-tailed distributions into symmetric ($t_3$-like) and asymmetric (mixture). Polychoric is unbiased under normality and acquires a modest positive bias (up to $\approx +0.10$ at low $\rho$) under asymmetric heavy-tailed latents. Kendall's~$\tau$ has the most consistent shape across distributions.}
\label{fig:polychoric-nonnormal}
\end{figure}

Three patterns emerge.
First, polychoric correlation is far more robust to non-normality than the textbook caveat suggests: under right-skewed, left-skewed, symmetric heavy-tailed, and uniform latents, the median $\hat\rho_{\mathrm{poly}}$ stays within about $0.01$ of the truth across the full $\rho$ range.
Second, the exception is asymmetric-heavy latents (a 70/30 mixture of normals with unequal variances and a shift), where polychoric acquires a positive bias of up to $+0.10$ at low $\rho$ (about $+0.05$ at moderate $\rho$) while Pearson and Spearman inherit a milder version of the same upward drift at low $\rho$ (up to $+0.06$ and $+0.08$ at $\rho = 0.1$) before attenuating below the truth from $\rho \approx 0.4$ onward.
Third, Kendall's~$\tau$ shows the most consistent shape across distributions: a downward shift that grows with $\rho$, attributable to its different scale (in the population $\tau = \tfrac{2}{\pi}\arcsin\rho \neq \rho$ under normality) rather than to distributional sensitivity.
When latent normality is doubtful, Kendall's~$\tau$ is the safer choice because its behavior is the most predictable.

\noindent
\textbf{Practical implication:} When latent normality holds (e.g., the ordinal variable represents a Likert-scale measurement of an attitude), polychoric/polyserial correlation should be preferred.
When latent normality is doubtful (e.g., the ordinal variable represents inherently discrete levels like educational degrees), Kendall's~$\tau$ or Spearman's~$\rho$ are safer choices.
The \code{smart\_cor()} function exposes this choice via the \code{assume\_latent\_normal} argument.

How, then, should a practitioner assess whether the latent normality assumption is credible?
The literature offers several complementary approaches.
The foundational work of \citet{muthen1985comparison,muthen1992comparison} established that normal-theory estimators such as polychoric correlation are reasonably robust to mild departures from bivariate normality but can exhibit meaningful bias under severe departures: a pattern consistent with the simulation results in \autoref{fig:polychoric-nonnormal}.
\citet{flora2004empirical} formalized this robustness boundary empirically, showing that polychoric estimation remains largely unbiased when skewness stays below roughly $|1.0|$ and excess kurtosis below $|2.0|$; more extreme departures shift the estimator toward overestimation.

The \pkg{smartcor} package implements a formal test for bivariate normality based on \citet{joreskog2005structural}.
The \code{test\_bivariate\_normality()} function fits the polychoric model by maximum likelihood and extracts the likelihood-ratio (LR) chi-square statistic that compares the observed cell frequencies of the two ordinal variables to the frequencies expected under a bivariate normal threshold model.
Specifically, given two ordinal variables with $m_1$ and $m_2$ categories, the maximum likelihood (ML) procedure jointly estimates $m_1 - 1$ row thresholds, $m_2 - 1$ column thresholds, and the latent correlation $\rho$, for a total of $m_1 + m_2 - 1$ parameters.\footnote{The $m_1 - 1$ row thresholds and $m_2 - 1$ column thresholds account for $m_1 + m_2 - 2$ parameters, and the latent correlation $\rho$ is the remaining one, giving $(m_1 - 1) + (m_2 - 1) + 1 = m_1 + m_2 - 1$ in total.}
The contingency table has $m_1 \times m_2 - 1$ free frequencies (the total is fixed), so the test has $\mathrm{df} = m_1 m_2 - m_1 - m_2$ degrees of freedom.
A significant result (small $p$-value) indicates that the bivariate normal threshold model is a poor fit to the observed data, suggesting that polychoric correlation may be biased and a distribution-free alternative such as Kendall's~$\tau$ should be preferred.
For a $2 \times 2$ table, $\mathrm{df} = 0$ and the model is saturated and no test is possible.

The \code{assume\_latent\_normal = "auto"} option in \code{smart\_cor()} leverages this test to make a data-driven choice: it runs \code{test\_bivariate\_normality()} internally and selects polychoric correlation when the test does not reject the possibility of the scalar random variable coming from a latent normal distribution, or Kendall's~$\tau$ when the latent normality test is rejected.
Thus, \code{smartcor} shifts the latent normality decision from a subjective judgment to an automated evidence-based recommendation.

\subsection{Practical recommendations} \label{sec:recommendations}

\autoref{tab:recommendations} summarizes the recommended methods by variable-type combination and encodes the default decision logic of \code{smart_cor()}. 

\begin{table}[!h]
\centering
\caption{Recommended methods by variable-type combination, as implemented in \code{smart\_cor()}. The lower block contains association measures (not correlation; see \autoref{sec:when}).}
\label{tab:recommendations}
{\scriptsize{
\begin{tabular}{lll}
\toprule
Variable types & Default method & Alternative \\
\midrule
\multicolumn{3}{l}{\emph{Correlation methods (signed, $[-1, +1]$)}} \\
Continuous + continuous & Pearson & Spearman (suggested, especially if nonlinear) \\
Continuous + binary & Pearson = Point-biserial & --- \\
Binary + binary & Pearson = Phi & Tetrachoric (preferred, if latent normality holds) \\
Continuous + ordinal & Spearman/Kendall & Polyserial (preferred, if latent normality holds) \\
Ordinal + ordinal & Kendall's $\tau$ & Polychoric (preferred, if latent normality holds) \\
Binary + ordinal & Rank-biserial & Spearman \\
\midrule
\multicolumn{3}{l}{\emph{Association measures (unsigned, $[0, 1]$)}} \\
Continuous + categorical & Cram\'{e}r's $V$ (binned) & --- \\
Binary + categorical & Cram\'{e}r's $V$ & --- \\
Ordinal + categorical & Cram\'{e}r's $V$ & --- \\
Categorical + categorical & Cram\'{e}r's $V$ & Theil's $U$, Tschuprow's $T$ \\
\bottomrule
\end{tabular}
}}
\end{table}

The upper block lists correlation methods, applicable when both variables have at least ordinal measurement; the lower block lists association measures, applicable when at least one variable is nominal categorical.
The defaults are deliberately assumption-light: none of the recommended methods in the second column requires a latent-normality assumption. 
The latent-variable alternatives (tetrachoric, polyserial, polychoric) deliver tighter estimates when the underlying continuum is genuinely Gaussian, but the assumption is untestable from the discrete data alone in any strict sense and unwarranted assumptions can bias the resulting coefficient. 
To strike a balance, \code{smart\_cor()} performs a likelihood ratio test for bivariate latent normality \citep{joreskog2005structural} on each ordinal pair and switches to the latent-variable alternative only when the test fails to reject normality at the user-specified level (default $\alpha = 0.05$); otherwise it retains the rank-based or contingency-based default.
Binary--binary pairs are the exception: a $2 \times 2$ table is saturated by the two thresholds and the correlation, leaving no degrees of freedom to test latent normality, so \code{smart\_cor()} reports $\phi$ and reserves tetrachoric correlation for users who assert the latent-normal model explicitly (\code{assume\_latent\_normal = TRUE} or \code{method = "tetrachoric"}).
Failure to reject is of course not equivalent to accepting normality, particularly at small sample sizes, and users who can defend the latent-normal model on substantive grounds, or who prefer a different testing threshold, can override the selection through the method argument.

\subsection{Usage of correlations in existing literature} 
\label{sec:literature}

The preceding sections describe the ideal decision procedure for correlation calculations.
For a given pair of variables, an analyst determines the measurement scale of each, identifies the type-appropriate estimator from \autoref{tab:recommendations}, and (where that estimator rests on a latent-variable model) decides whether the latent-normal assumption is defensible before computing the coefficient and a corresponding confidence interval.
Consequently, it is fair to ask how often that procedure is actually carried out in published empirical work, and whether the default choices of existing software affect the results and conclusions.
To answer both questions, we conducted a study of correlation usage in the data supplements of published economics articles, using the coefficients researchers computed themselves as the benchmark and our \code{smart\_cor()} as the comparison.

Source of research papers and supplementary materials in our mini-study is the \emph{Find Economic Articles with Data} (EJD) database maintained by \citet{kranz2023ejd}, which indexes articles from the American Economic Association journals, the \emph{Review of Economic Studies}, the \emph{Review of Economics and Statistics}, and smaller samples from the \emph{Quarterly Journal of Economics}, the \emph{Journal of Political Economy}, \emph{Econometrica}, and others, together with their deposited replication supplements.
\autoref{tab:ejd-funnel} traces the corpus from that index to the set of coefficients we analyze.
Of the $11{,}889$ indexed articles, $9{,}312$ provide a data supplement.
We specially look at the subset of $1{,}101$ studies whose supplements are  retrievable through an open programmatic interface (Harvard Dataverse and Zenodo); the remainder are hosted behind logins (openICPSR) or on legacy links that now redirect to them.
Of these, $1{,}077$ were successfully retrieved and screened, $780$ contained analysis code, and $197$ ($25.3\%$ of those with code) contained at least one correlation command.
Each of the $197$ was examined in full: every correlation call was traced to the data in the supplementary files and, whenever possible, correlation values were recomputed to reproduce the paper's outputs.
Eighty papers' correlations were recomputable based on this criterion; the $117$ that were not divide into $76$ resting on restricted or proprietary microdata (administrative registers, confidential vendor panels, etc.), $25$ whose only correlations are internal to simulations, $6$ that occur in bundled third-party library code, and $10$ with other technical reasons.
The recomputation yielded $3{,}364$ distinct variable pairs, of which $3{,}322$ have at least $15$ complete observations and form the analysis set; pairs below that size are excluded because type detection is unreliable when a continuous variable happens to take few distinct values (\autoref{sec:detect}).

\begin{table}[h]
\centering
\caption{From the EJD index to the analyzed coefficients. Each row is a subset of the one above it.}
\label{tab:ejd-funnel}
\begin{tabular}{lr}
\toprule
Stage & Count \\
\midrule
Articles indexed in EJD & $11{,}889$ \\
\quad with a data supplement & $9{,}312$ \\
\quad supplement anonymously retrievable (sampling frame) & $1{,}101$ \\
\quad retrieved and screened & $1{,}077$ \\
\quad containing analysis code & $780$ \\
\quad issuing $\geq 1$ correlation command & $197$ \\
\midrule
Papers recomputable (data shipped, light preparation) & $80$ \\
Papers contributing recomputed pairs & $79$ \\
Distinct variable pairs recomputed & $3{,}364$ \\
Pairs analyzed ($n \geq 15$) & $3{,}322$ \\
\bottomrule
\end{tabular}
\end{table}

The first finding concerns the method itself, before any question of whether it is the right one.
Across the $197$ papers, $1{,}124$ of the $1{,}218$ correlation calls ($92.3\%$) are Pearson (identified using the Stata \code{correlate}/\code{pwcorr}, \proglang{R} \fct{cor}, \proglang{MATLAB} \fct{corrcoef}, and \proglang{Python} defaults) with Spearman ($63$), Kendall ($17$), partial Pearson ($10$), and polychoric ($4$) making up the remainder.
At the paper level, $174$ of the $197$ ($88\%$) use Pearson and nothing else.
This is unconditional on the variables being correlated: the default is applied whether the pair consists of two continuous variables, a continuous variable and a binary indicator, or two ordinal Likert items.

For each recomputed pair, we computed the authors' own coefficient in \proglang{R}, reproducing their published value where a public version of the paper allowed the check, so that the comparison is like--for-like—alongside the coefficient \code{smart\_cor()}, which selects on the identical observations.
Of the $3{,}322$ analyzed pairs, $2{,}381$ ($71.7\%$) fall in what we term the Pearson-equivalent class: the estimator \code{smart\_cor()} selects is either Pearson itself or one of its two algebraic identities, the point-biserial coefficient for a continuous--binary pair or $\phi$ for a binary--binary pair (\autoref{sec:equiv}), which return the same number.
The remaining $941$ instances diverge.
This high equivalence rate should not be read as evidence that Pearson is broadly appropriate.
Of the ten variable-type pairings enumerated in \autoref{tab:recommendations}, Pearson (or one of its two identities) is the type-appropriate estimator in only three: continuous--continuous, continuous--binary, and binary--binary.
The equivalence rate is high because the corpus is compositionally dominated by exactly those three pairings: economists overwhelmingly correlate continuous magnitudes and binary indicators, for which Pearson is often the near-best tool.
Indeed, not a single variable in the corpus enters a correlation as a nominal categorical with its $k$ levels intact: of the $761$ distinct variables behind the analyzed pairs, type detection classifies $506$ as continuous, $121$ as binary, $79$ as ordinal, and $55$ as count, and none as categorical---where group membership matters, it appears only as already-constructed binary indicator variables.
The decomposition by pairing makes the point sharply.
Among the $2{,}590$ analyzed pairs in which no variable is ordinal, the method diverges in $8.1\%$ of cases---the residual arising almost entirely where authors deliberately chose a rank method on continuous data.
Among the $732$ pairs in which at least one variable is ordinal, the method diverges in \emph{every single case} ($732$ of $732$, $100\%$).
The choice of correlation method is therefore not a matter that ``usually does not matter''; it does not matter precisely where the variables are continuous, and it almost always matters the moment an ordinal variable enters.

In the three pairings where the coefficient coincides with Pearson, the coincidence is exact and provable (\autoref{app:proofs}), but the literature does not arrive at it by the reasoning that would establish it.
Recognizing that a continuous--binary or binary--binary pairing makes Pearson the correct estimator requires first classifying each variable's measurement scale and identifying the pairing; only then does one know that the point-biserial or $\phi$ formula (and hence the generic Pearson call) returns the type-appropriate coefficient.
In the corpus, that coefficient is instead reached generically: a single \code{correlate} or \code{pwcorr} statement in a software package is applied across a block of variables, without a preceding step that establishes the pairing.
The resulting number may be correct, but its correctness is incidental rather than demonstrated, and is applied unchanged to the seven pairings for which it is not the type-appropriate choice.
This is not to characterize the generic call as an error where it happens to coincide; it is to observe that the reasoning distinguishing the three coincident pairings from the seven non-coincident ones is, in practice, rarely performed.

The value of automating the choice of correlation metric is clearest when the manual alternative is spelled out.
To select the type-appropriate coefficient for a single pair, an analyst would, in principle: 
(i) determine each variable's measurement scale, which is not always evident from the stored values---a $0/1$ column may encode a genuine dichotomy or a coarsened continuum, and a $1$--$5$ column may be a Likert attitude with a plausible latent continuum or a set of intrinsically discrete categories; 
(ii) look up the estimator appropriate to the resulting pairing; 
(iii) for any pairing whose recommended estimator is a latent-variable coefficient (tetrachoric, polyserial, polychoric), decide whether the latent-normal model is defensible, ideally through a formal test rather than an assumption (\autoref{sec:test-normality}), and fall back to a rank-based estimator when it is not; and 
(iv) compute a confidence interval derived for the chosen estimator rather than the Pearson interval.
Carried out by hand for every cell of every correlation table, across every variable pairing a paper reports, this is a substantial and repetitive cost, and a single trusted default is an understandable economy even if it may be wrong.

\code{smart\_cor()} performs all four steps automatically for each pair, at no marginal cost to the analyst: it detects the measurement type of each variable, selects the estimator by pairing, runs the latent-normality test where the pairing calls for one and switches estimators on its outcome, and returns an interval matched to the method.

Where the method diverges, the magnitude of the change ranges from negligible to consequential.
Among the $941$ divergent pairs the median absolute change in the coefficient is small, $|\Delta r| \approx 0.017$, but the distribution has a substantial tail, reaching $0.47$, and in $36$ pairs the divergence moves the coefficient across the conventional $p < 0.05$ threshold, so relationships that were rendered insignificant with default correlation methods in existing software packages, become significant with \code{smart\_cor}'s statistically-accurate choices.
The direction is not uniform, and this is worth stating plainly: in the ordinal pairings \code{smart\_cor()}'s latent-variable estimators recover more association than Pearson, consistent with the attenuation documented in \autoref{sec:attenuation}, whereas in the $209$ divergent continuous-data pairs the authors had chosen a rank method (typically Spearman on skewed magnitudes) and \code{smart\_cor()} defaults to Pearson where the package is the more conservative choice, and the difference reflects a genuine distinction between linear and monotone association rather than an error on either side.

\section[The smartcor package]{The \pkg{smartcor} package} \label{sec:package}

\subsection{Design philosophy} \label{sec:design}

The \pkg{smartcor} package is built around four principles:

\begin{enumerate}
  \item \textbf{Smart defaults.} The package ``just works'' for the common
    case: pass two variables or a data frame, and the appropriate methods
    are selected automatically.
  \item \textbf{Transparency.} Every decision is explained. Each result
    includes a \code{rationale} field describing why a particular method was
    chosen, and verbose mode prints the detection and selection process.
  \item \textbf{User override.} Users can force any method or type assignment
    via the \code{method}, \code{x\_type}, and \code{y\_type} arguments.
  \item \textbf{Tidy integration.} All output objects have \code{tidy()} methods
    (via \pkg{broom}; \citealp{R-broom}) returning tibbles, and
    \code{smart\_cormat()} produces matrices compatible with \pkg{ggplot2}
    \citep{R-ggplot2} visualization.
\end{enumerate}

The package can be installed from CRAN:

\begin{CodeInput}
R> install.packages("smartcor")
R> library("smartcor")
\end{CodeInput}

\noindent
The core dependencies are \pkg{cli} for formatted console output, \pkg{rlang}, \pkg{tibble}, and \pkg{generics} for tidyverse-style data handling, \pkg{polycor} \citep{R-polycor} for polychoric, polyserial, and tetrachoric correlations, and \pkg{ggplot2} \citep{R-ggplot2} for the heatmap visualizations of \autoref{sec:viz}.

\subsection{Variable type detection} \label{sec:detect}

The \code{detect\_type()} function classifies a variable into one of five types (continuous, binary, ordinal, count, or categorical) using the following rules, applied in order:

\begin{enumerate}
  \item Ordered factors $\rightarrow$ \code{"ordinal"}.
  \item Unordered factors or character vectors with exactly 2 unique values
    $\rightarrow$ \code{"binary"}; with 3 or more $\rightarrow$
    \code{"categorical"}.
  \item Logical vectors $\rightarrow$ \code{"binary"}.
  \item Numeric vectors with exactly 2 unique values $\rightarrow$
    \code{"binary"}.
  \item Numeric vectors with $\leq$~\code{ordinal\_threshold} (default: 10)
    unique values $\rightarrow$ \code{"ordinal"}.
  \item Numeric vectors that are non-negative integer-valued and have more
    than \code{ordinal\_threshold} unique values $\rightarrow$ \code{"count"}
    (disable with \code{detect\_count = FALSE}).
  \item All other numeric vectors $\rightarrow$ \code{"continuous"}.
\end{enumerate}

\noindent
Rule~6 is a syntactic heuristic: any non-negative, integer-valued vector with more than \code{ordinal\_threshold} distinct values is labeled \code{"count"}, including integer-recorded \emph{measurements} such as age in years or engine horsepower that do not, strictly, arise from counting.
Since \pkg{smartcor} treats count variables identically to continuous variables for every correlation, this labeling never changes the selected method or the estimate; it affects only the descriptive type reported.
Set \code{detect\_count = FALSE} to classify such variables as \code{"continuous"}.

\begin{CodeInput}
R> detect_type(rnorm(100))
\end{CodeInput}
\begin{CodeOutput}
[1] "continuous"
\end{CodeOutput}

\begin{CodeInput}
R> detect_type(c(0, 1, 0, 1, 1, 0))
\end{CodeInput}
\begin{CodeOutput}
[1] "binary"
\end{CodeOutput}

\begin{CodeInput}
R> detect_type(c(1, 2, 3, 4, 5, 3, 2, 4))
\end{CodeInput}
\begin{CodeOutput}
[1] "ordinal"
\end{CodeOutput}

\begin{CodeInput}
R> detect_type(c("red", "blue", "green"))
\end{CodeInput}
\begin{CodeOutput}
[1] "categorical"
\end{CodeOutput}

\noindent
Full argument documentation for \code{detect\_type()}, and for every function introduced below, is available in the package help pages.

\subsection[Core function: smart\_cor()]{Core function: \code{smart\_cor()}} \label{sec:smartcor}

The primary function \code{smart\_cor(x, y)} computes the best-suited correlation for two variables.
It detects types, selects a method, and returns a richly annotated result.

\begin{CodeInput}
R> result <- smart_cor(mtcars$mpg, mtcars$wt, verbose = FALSE)
R> result
\end{CodeInput}
\begin{CodeOutput}
-- Smart Correlation --
Estimate: -0.8677
Method: Pearson Correlation
Variables: mtcars$mpg ("continuous") x mtcars$wt ("continuous")
N: 32
p-value: < 0.001
Test: t-test on r (cor.test)
H0: rho = 0
Small p-values (e.g., p < 0.05) indicate evidence against H0.
95\% CI: [-0.9338, -0.7441]
Source: Fisher z (cor.test)

Both variables are continuous; Pearson correlation selected.

i Alternatives: "spearman" and "kendall" (pass `method = "..."` to use)
\end{CodeOutput}

\noindent
The result object is a list of class \code{"smartcor"} with elements covering identifiers (\code{x\_name}, \code{y\_name}, \code{x\_type}, \code{y\_type}), the point estimate (\code{estimate}, \code{method}, \code{method\_label}), inference (\code{statistic}, \code{p.value}, \code{p\_method}, \code{null\_hypothesis}, \code{p\_interpretation}, \code{ci\_lower}, \code{ci\_upper}, \code{conf\_level}, \code{ci\_method}, \code{ci\_source}), the rationale text, and sample sizes.
The \code{tidy()} method returns a one-row tibble with twenty columns; transposing it with \code{t()} makes every column readable in print:

\begin{CodeInput}
R> library("broom")
R> t(tidy(result))
\end{CodeInput}
\begin{CodeOutput}
                 [,1]
var_x            "mtcars$mpg"
var_y            "mtcars$wt"
x_type           "continuous"
y_type           "continuous"
estimate         "-0.8676594"
method           "pearson"
method_label     "Pearson Correlation"
rationale        "Both variables are continuous; Pearson correlation selected."
statistic        "-9.559044"
p.value          "1.293959e-10"
p_method         "t-test on r (cor.test)"
null_hypothesis  "rho = 0"
p_interpretation "Small p-values (e.g., p < 0.05) indicate evidence against H0."
ci_lower         "-0.9338264"
ci_upper         "-0.7440872"
conf_level       "0.95"
ci_method        "analytic"
ci_source        "Fisher z (cor.test)"
n                "32"
n_complete       "32"
\end{CodeOutput}

\noindent
For a pair involving different types, \code{smart\_cor()} selects the appropriate specialized method.
The default \code{assume_latent_normal = "auto"} runs the bivariate normality test only when at least one data type is ordinal to test for latent normality, for instance, in ordinal-ordinal pair one needs to choose between Kendall and Polychoric correlations.
For ordinal--ordinal pairs, the default \code{"auto"} mode runs the bivariate normality test to decide between polychoric and Kendall.
Users can override with \code{assume\_latent\_normal = FALSE}:

\begin{CodeInput}
R> smart_cor(mtcars$gear, mtcars$carb,
+    assume_latent_normal = FALSE, verbose = FALSE)
\end{CodeInput}
\begin{CodeOutput}
-- Smart Correlation --
Estimate: 0.0980
Method: Kendall's Tau-b
Variables: mtcars$gear ("ordinal") x mtcars$carb ("ordinal")
N: 32
p-value: 0.5302
Test: asymptotic normal (cor.test, exact = FALSE)
H0: tau = 0
Small p-values (e.g., p < 0.05) indicate evidence against H0.
95% CI: [-0.1455, 0.3303]
Source: Fieller, Hartley, and Pearson (1957)

Both variables are ordinal. Kendall's tau selected (no latent normality
assumption; stable bias across distributions).

i Alternatives: "polychoric", "spearman", and "gamma"
  (pass `method = "..."` to use)
\end{CodeOutput}

\noindent
The default behavior of \code{smart\_cor()} can be summarized as follows.
Variable types are auto-detected via \code{detect\_type()} unless the user supplies \code{x\_type} or \code{y\_type}.
For ordinal--ordinal and ordinal--continuous pairs, the latent normality assumption is tested automatically (\code{assume\_latent\_normal = "auto"}) using a likelihood-ratio chi-square test (\autoref{sec:test-normality}); if the test rejects, a distribution-free method is selected instead.
For binary--binary pairs the $2 \times 2$ table is saturated and cannot support this test, so the default is the $\phi$ coefficient; tetrachoric correlation is available through \code{assume\_latent\_normal = TRUE} or \code{method = "tetrachoric"}.
Confidence intervals and $p$-values are computed analytically wherever a closed-form or asymptotic source is available; this is the default for every method except Theil's $U$, whose CI is obtained by a percentile bootstrap and $p$-value is obtained using a permutation test against $H_0: X \perp Y$ (both $B = 500$, set via \code{n\_boot} and \code{n\_perm}).
\autoref{tab:inference-sources} documents the estimate, $p$-value, and confidence-interval source used for each correlation method, with citations.
The \code{bootstrap = "auto"} (default) ensures that bootstrapping is performed only when an analytic confidence interval is unavailable or fails to compute (e.g., singular Hessian, sparse $2 \times 2$ table); setting \code{bootstrap = TRUE} forces bootstrap intervals for all methods, useful when cross-checking analytic intervals.
The ordinary paired bootstrap is centred at the observed estimate rather than the null, so it is not used to construct $p$-values; analytic tests remain in place, and Theil's~$U$ retains its permutation test.
Data pairs with at least one missing value are removed by \code{ignore\_na = TRUE}.

\paragraph{Sanity check on simulated data.}
A two-line sanity check makes the behaviour easier to internalise. Simulating two correlated normals with $\rho = 0.7$, \code{smart\_cor()} reports a Pearson estimate near the truth with an analytic Fisher-$z$ confidence interval. Replacing the variables with their ranks reproduces the well-known identity that Pearson on ranks equals Spearman on the originals.

\begin{CodeInput}
R> set.seed(7)
R> xy <- MASS::mvrnorm(500, mu = c(0, 0),
+                       Sigma = matrix(c(1, 0.7, 0.7, 1), 2, 2))
R> x <- xy[, 1]; y <- xy[, 2]
R> smart_cor(x, y, verbose = FALSE)
\end{CodeInput}
\begin{CodeOutput}
-- Smart Correlation --
Estimate: 0.7167
Method: Pearson Correlation
Variables: x ("continuous") x y ("continuous")
N: 500
p-value: < 0.001
Test: t-test on r (cor.test)
H0: rho = 0
Small p-values (e.g., p < 0.05) indicate evidence against H0.
95% CI: [0.6712, 0.7568]
Source: Fisher z (cor.test)

Both variables are continuous; Pearson correlation selected.

i Alternatives: "spearman" and "kendall" (pass `method = "..."` to use)
\end{CodeOutput}

\noindent
The same call on the rank-transformed pair returns the Spearman of the original pair exactly:

\begin{CodeInput}
R> all.equal(cor(rank(x), rank(y)),
+             cor(x, y, method = "spearman"))
\end{CodeInput}
\begin{CodeOutput}
[1] TRUE
\end{CodeOutput}

\noindent
\code{rank(x)} and \code{rank(y)} are integer-valued with many unique values, so \code{detect\_type()} flags them as counts; under the smartcor convention that counts are treated as continuous numerics, the function returns Pearson correlation, whose value coincides with the Spearman of the originals by construction, which is a useful end-to-end sanity check before running the same pipeline on real data.

\begin{table}[!h]
\centering
\caption{Source of the point estimate, $p$-value, and confidence interval for each correlation method in \pkg{smartcor}. Cells marked $^{\ast}$ are computed natively by \pkg{smartcor} rather than delegated to base \proglang{R} or an external package; for the tetrachoric, polychoric, and polyserial methods, the Wald $p$-value and Fisher-$z$ interval are constructed by \pkg{smartcor} from the maximum-likelihood standard error returned by \pkg{polycor}. Bootstrap confidence intervals are available via \code{bootstrap = TRUE} and serve as automatic fallback when an analytic interval fails.}
\label{tab:inference-sources}
\scriptsize
\begin{tabular}{p{1.9cm}p{4.4cm}p{2.5cm}p{4.5cm}}
\toprule
Correlation & Estimate & $p$-value & Confidence interval \\
\midrule
Pearson          & \code{stats::cor.test} (default Pearson) & \code{cor.test} (analytic, $t$) & \code{cor.test} (Fisher $z$) \\
Spearman         & \code{stats::cor.test( method="spearman", exact=FALSE)} & \code{cor.test} (asymptotic AS~89) & \citet{bonett2000sample}, Fisher $z$, $\mathrm{SE}_z = \sqrt{(1+r^2/2)/(n-3)}$$^{\ast}$ \\
Kendall          & \code{stats::cor.test( method="kendall", exact=FALSE)} & \code{cor.test} (asymptotic normal) & \citet{fieller1957tests}, Fisher $z$, $\mathrm{SE}_z = \sqrt{0.437/(n-4)}$$^{\ast}$ \\
Point-Biserial   & \code{stats::cor.test} on continuous and 0/1-coded binary & \code{cor.test} (analytic, $t$) & \code{cor.test} (Fisher $z$) \\
Phi              & $2 \times 2$ contingency formula$^{\ast}$ & \code{chisq.test} (\code{correct=FALSE}) & Noncentral $\chi^2$ pivot (signed)$^{\ast}$ \\
Tetrachoric      & \code{polycor::polychor(ML=TRUE, std.err=TRUE)} on the $2 \times 2$ table & Wald, Fisher $z$$^{\ast}$ & \citet{olsson1979maximum} and \citet{joreskog1994estimation}, Fisher $z$$^{\ast}$ \\
Polychoric       & \code{polycor::polychor(ML=TRUE, std.err=TRUE)} & Wald, Fisher $z$$^{\ast}$ & \citet{olsson1979maximum} and \citet{joreskog1994estimation}, Fisher $z$$^{\ast}$ \\
Polyserial       & \code{polycor::polyserial(ML=TRUE, std.err=TRUE)} & Wald, Fisher $z$$^{\ast}$ & \citet{olsson1982polyserial}, Fisher $z$$^{\ast}$ \\
Cram\'{e}r's $V$ & $\sqrt{\chi^2/(n(k-1))}$ from \code{chisq.test}$^{\ast}$ & \code{chisq.test} (analytic) & Noncentral $\chi^2$ pivot$^{\ast}$ \\
Tschuprow's $T$  & $\sqrt{\chi^2/(n\sqrt{(r-1)(c-1)})}$ from \code{chisq.test}$^{\ast}$ & \code{chisq.test} (analytic) & Noncentral $\chi^2$ pivot (rescaled $V$)$^{\ast}$ \\
Yule's $Q$       & $(ad - bc)/(ad + bc)$ on $2 \times 2$$^{\ast}$ & Wald, Fisher $z$$^{\ast}$ & \citet{brown1977sampling} (gamma on $2 \times 2$)$^{\ast}$ \\
Goodman--Kruskal $\gamma$ & $(P - Q)/(P + Q)$ from concordant/discordant pairs$^{\ast}$ & Wald, Fisher $z$$^{\ast}$ & \citet{brown1977sampling} (closed-form SE)$^{\ast}$ \\
Rank-biserial    & Glass formula $2(\bar{r}_1 - \bar{r}_0)/n$$^{\ast}$ & Wald, Fisher $z$$^{\ast}$ & \citet{cliff1996ordinal} asymptotic (Cliff's $\delta$)$^{\ast}$ \\
Theil's $U$      & Entropy-based $U(Y|X)$$^{\ast}$ & Permutation, $B = 500$$^{\ast}$ & Bootstrap, $B = 500$$^{\ast}$ \\
\bottomrule
\end{tabular}
\end{table}

\subsection[Testing latent normality: test\_bivariate\_normality()]{Testing latent normality: \code{test\_bivariate\_normality()}} \label{sec:test-normality}

As discussed in \autoref{sec:normality}, polychoric and tetrachoric correlations assume that the observed ordinal or binary variables arise from an underlying bivariate normal distribution.
Following the approach of \citet{joreskog2005structural}, our \code{test\_bivariate\_normality()} function provides a formal likelihood-ratio test of the latent normality assumption.
The implementation fits the polychoric model by maximum likelihood using \code{polycor::polychor(tab, ML = TRUE, std.err = TRUE)}, which jointly estimates row thresholds, column thresholds, and the latent correlation $\rho$.
It then extracts the likelihood-ratio chi-square statistic comparing the observed cell frequencies to those expected under the fitted bivariate normal threshold model (see \autoref{sec:normality} for the degrees of freedom).
A small $p$-value indicates that the bivariate normal model is a poor fit, suggesting that polychoric correlation may be biased and Kendall's~$\tau$ should be preferred:

\begin{CodeInput}
R> set.seed(42)
R> z <- MASS::mvrnorm(500, mu = c(0, 0),
+    Sigma = matrix(c(1, 0.5, 0.5, 1), 2, 2))
R> x <- cut(z[, 1], breaks = qnorm(c(0, 0.2, 0.4, 0.6, 0.8, 1)),
+    labels = FALSE)
R> y <- cut(z[, 2], breaks = qnorm(c(0, 0.2, 0.4, 0.6, 0.8, 1)),
+    labels = FALSE)
R> test_bivariate_normality(x, y)
\end{CodeInput}
\begin{CodeOutput}
-- Bivariate Normality Test (LR Chi-Square) --

Table dimensions: 5 x 5
N = 500
Polychoric rho: 0.4399
LR chi-square: 22.2181 (df = 15, p = 0.1022)
Conclusion: no evidence against bivariate normality at alpha = 0.05. Polychoric
correlation is appropriate.
\end{CodeOutput}

\noindent
Because \code{assume\_latent\_normal = "auto"} is the default, this test is integrated directly into the standard \code{smart\_cor()} workflow.
The output includes the LR test result alongside the correlation estimate:

\begin{CodeInput}
R> smart_cor(x, y, verbose = FALSE)
\end{CodeInput}
\begin{CodeOutput}
-- Smart Correlation --
Estimate: 0.4406
Method: Polychoric Correlation
Variables: x ("ordinal") x y ("ordinal")
N: 500
p-value: < 0.001
Test: Wald test on Fisher-z scale
H0: rho = 0
Small p-values (e.g., p < 0.05) indicate evidence against H0.
95% CI: [0.3565, 0.5176]
Source: Olsson (1979) and Joreskog (1994)

Latent normality: LR chi-sq = 22.22, df = 15, p = 0.102 (not rejected)
Test: LR chi-square test of latent bivariate normality
H0: (x, y) ~ bivariate normal
Small p-values indicate evidence against latent bivariate normality.

Both variables are ordinal. Polychoric correlation selected under the
assumption of latent bivariate normality.

i Alternatives: "kendall", "spearman", and "gamma"
  (pass `method = "..."` to use)
\end{CodeOutput}

\noindent
The two test results sit side by side. The first $p$-value ($p < 0.001$) tests the polychoric estimate itself against the independence null $H_0$:~$\rho = 0$, using a Wald statistic on the Fisher-$z$ scale: $z = \mathrm{atanh}(\hat\rho) / (\widehat{\mathrm{SE}}(\hat\rho) / (1 - \hat\rho^2))$. The second test (latent normality, $p = 0.102$) compares the polychoric model against the saturated multinomial via likelihood ratio; large values reject the bivariate-normal threshold model. Because that test did \emph{not} reject ($p = 0.102$), \code{smart\_cor()} retained polychoric correlation; had it rejected, the function would have defaulted to Kendall's~$\tau_b$. The confidence interval uses the Olsson--J\"oreskog standard error of $\hat\rho$ on the Fisher~$z$ scale (\autoref{tab:inference-sources}); bootstrap inference is no longer needed but remains available via \code{bootstrap = TRUE} as a sanity check.

\subsection[Correlation matrix: smart\_cormat()]{Correlation matrix: \code{smart\_cormat()}} \label{sec:cormat}

For data frames with mixed variable types, \code{smart\_cormat()} computes pairwise correlations using the appropriate method for each pair. We illustrate on the entire \code{mtcars} data frame with all eleven columns and no custom changes. Auto-detection (the central premise of \pkg{smartcor}) handles the type assignment internally:

\begin{CodeInput}
R> mat <- smart_cormat(mtcars, assume_latent_normal = FALSE,
+                      verbose = FALSE)
R> mat$types
\end{CodeInput}
\begin{CodeOutput}
         mpg          cyl         disp           hp         drat
"continuous"    "ordinal" "continuous"      "count" "continuous"
          wt         qsec           vs           am         gear
"continuous" "continuous"     "binary"     "binary"    "ordinal"
        carb
   "ordinal"
\end{CodeOutput}

\noindent
\code{detect\_type()} correctly recognizes \code{cyl}, \code{gear}, and \code{carb} as ordinal (small numeric ranges), \code{vs} and \code{am} as binary (two unique values), and \code{hp} as a count (non-negative integers exceeding \code{ordinal\_threshold}); the remaining five columns are continuous. Counts are treated as numeric continuous variables for purposes of method selection, so a count $\times$ continuous pair routes to Pearson and a count $\times$ ordinal pair routes through the same polyserial-or-Spearman logic as continuous $\times$ ordinal. The correlation matrix follows directly:

\begin{CodeInput}
R> signif(mat$correlations, 2)
\end{CodeInput}
\begin{CodeOutput}
       mpg   cyl  disp    hp   drat    wt   qsec    vs     am   gear   carb
mpg   1.00 -0.91 -0.85 -0.78  0.680 -0.87  0.420  0.66  0.600  0.540 -0.660
cyl  -0.91  1.00  0.93  0.90 -0.680  0.86 -0.570 -0.88 -0.570 -0.510  0.470
disp -0.85  0.93  1.00  0.79 -0.710  0.89 -0.430 -0.71 -0.590 -0.590  0.540
hp   -0.78  0.90  0.79  1.00 -0.450  0.66 -0.710 -0.72 -0.240 -0.330  0.730
drat  0.68 -0.68 -0.71 -0.45  1.000 -0.71  0.091  0.44  0.710  0.740 -0.130
wt   -0.87  0.86  0.89  0.66 -0.710  1.00 -0.170 -0.55 -0.690 -0.680  0.500
qsec  0.42 -0.57 -0.43 -0.71  0.091 -0.17  1.000  0.74 -0.230 -0.150 -0.660
vs    0.66 -0.88 -0.71 -0.72  0.440 -0.55  0.740  1.00  0.170  0.300 -0.710
am    0.60 -0.57 -0.59 -0.24  0.710 -0.69 -0.230  0.17  1.000  0.870 -0.073
gear  0.54 -0.51 -0.59 -0.33  0.740 -0.68 -0.150  0.30  0.870  1.000  0.098
carb -0.66  0.47  0.54  0.73 -0.130  0.50 -0.660 -0.71 -0.073  0.098  1.000
\end{CodeOutput}

\noindent
Two significant figures keep the matrix on a single block; \code{print(mat, digits = 3)} restores three-figure rounding when greater precision is needed. The result contains four elements: \code{\$correlations} (a numeric matrix), \code{\$methods} (a character matrix of method codes), \code{\$types} (detected types), and \code{\$details} (a list of full \code{smartcor} objects for each pair).
Every cell, including the diagonal, is computed by an explicit method (the diagonal records, e.g., \code{phi} for binary-binary, \code{kendall} for ordinal-ordinal, etc. as appropriate).
The \code{tidy()} method returns a long-format tibble with one row per variable pair.

\paragraph{Comparison against base \code{cor()}.}
Calling \code{cor(mtcars)} silently treats every column as continuous, returning Pearson everywhere. To see where that silent treatment makes a real difference, we filter the cell-wise difference $\hat\rho_{\mathrm{smart}} - \hat\rho_{\mathrm{Pearson}}$ to entries with $|\Delta| > 0.05$. By construction this filter excludes every pair for which Pearson is the equivalent estimator (continuous--continuous, continuous--binary, and continuous--count), since smartcor returns the same numerical value there; only pairs touching an ordinal variable (\code{cyl}, \code{gear}, \code{carb}) survive:

\begin{CodeInput}
R> base_pearson <- cor(mtcars)
R> diff_mat <- mat$correlations - base_pearson
R> ix <- which(abs(diff_mat) > 0.05 & upper.tri(diff_mat), arr.ind = TRUE)
R> data.frame(
+    pair     = paste(rownames(diff_mat)[ix[,1]],
+                     colnames(diff_mat)[ix[,2]], sep = "-"),
+    smartcor = round(mat$correlations[ix], 2),
+    pearson  = round(base_pearson[ix], 2),
+    diff     = round(diff_mat[ix], 2)
+  )
\end{CodeInput}
\begin{CodeOutput}
        pair smartcor pearson  diff
1    mpg-cyl    -0.91   -0.85 -0.06
2     cyl-hp     0.90    0.83  0.07
3     cyl-wt     0.86    0.78  0.08
4     cyl-vs    -0.88   -0.81 -0.07
5   mpg-gear     0.54    0.48  0.06
6    hp-gear    -0.33   -0.13 -0.21
7    wt-gear    -0.68   -0.58 -0.09
8  qsec-gear    -0.15   -0.21  0.06
9    vs-gear     0.30    0.21  0.10
10   am-gear     0.87    0.79  0.08
11  mpg-carb    -0.66   -0.55 -0.11
12  cyl-carb     0.47    0.53 -0.06
13 disp-carb     0.54    0.39  0.14
14   wt-carb     0.50    0.43  0.07
15   vs-carb    -0.71   -0.57 -0.14
16   am-carb    -0.07    0.06 -0.13
17 gear-carb     0.10    0.27 -0.18
\end{CodeOutput}

\noindent
Every row that crosses the threshold contains at least one ordinal column. 
The largest gap is at \code{hp}--\code{gear} where base \code{cor()} reports $r = -0.13$ while the type-aware Spearman reports $r = -0.33$; \code{cor(mtcars)} silently treats \code{gear} as continuous and produces the smaller, attenuated number. 
Variable pairs where Pearson is the equivalent estimator (continuous--continuous, continuous--binary, and, with counts treated as continuous, continuous--count and count--count), \code{smart\_cormat()} produces the same numerical value, but the \code{methods} matrix records the conceptually-correct estimator name. 
A reader of the output can tell which cells happen to coincide with Pearson and which were a substantive choice; \code{cor(mtcars)} cannot make this distinction.

\subsection[Method comparison: compare\_methods()]{Method comparison: \code{compare\_methods()}} \label{sec:compare}

The \code{compare\_methods()} function computes all applicable methods for a given variable pair, allowing users to assess how sensitive the result is to method choice.

\begin{CodeInput}
R> compare_methods(mtcars$gear, mtcars$carb, verbose = FALSE)
\end{CodeInput}
\begin{CodeOutput}
-- Method Comparison --
Variables: mtcars$gear ("ordinal") x mtcars$carb ("ordinal")
N: 32

Polychoric Correlation <- recommended
r = 0.2451 p = 0.2490 95% CI [-0.1734, 0.5886]
Kendall's Tau-b
r = 0.0980 p = 0.5302 95% CI [-0.1455, 0.3303]
Spearman Rank Correlation
r = 0.1149 p = 0.5312 95% CI [-0.2447, 0.4467]
Goodman-Kruskal's Gamma
r = 0.1405 p = 0.5433 95% CI [-0.3046, 0.5353]
\end{CodeOutput}

\noindent
The arrow (\code{<- recommended}) indicates the default method that \code{smart\_cor()} would select.
Here polychoric is recommended because \code{assume_latent_normal} defaults to \code{"auto"}, the latent-normality test supports normality for this pair, and the data therefore meet polychoric's underlying assumption.
All four methods report analytic confidence intervals: Kendall via Fieller--Hartley--Pearson, Spearman via Bonett--Wright, polychoric via Olsson/J\"oreskog (Fisher~$z$ scale), and Goodman--Kruskal~$\gamma$ via Brown--Benedetti (\autoref{tab:inference-sources}).
The widely different CI widths reflect the different precision properties of each estimator at $n = 32$, not a switch between bootstrap and analytic inference.

\subsection{Visualization} \label{sec:viz}

The \code{ggcor\_heatmap()} and \code{ggcor\_method\_heatmap()} functions produce \pkg{ggplot2}-based heatmaps from a \code{smartcormat} object.
The correlation heatmap displays the estimated values with a diverging color scale, while the method heatmap shows which correlation method was used for each pair, making the analytic heterogeneity of a mixed-type dataset immediately visible.

\begin{CodeInput}
R> library("ggplot2")
R> ggcor_heatmap(mat)
R> ggcor_method_heatmap(mat)
\end{CodeInput}

\noindent
A base \proglang{R} alternative is also available via \code{plot(mat)}, although it is not as flexible as the \code{ggplot} due to inherent design of \code{ggplot}.

\subsection[Pipeline-friendly interface: smart\_cor\_df()]{Pipeline-friendly interface: \code{smart\_cor\_df()}} \label{sec:cordf}

A central design choice in \pkg{smartcor} is that every function returns an S3 object with a consistent shape, and every such object has a \code{tidy()} method \citep{R-broom} producing a one-row-per-pair tibble.
This composes naturally with the tidyverse \citep{R-tidyverse}, the native pipe \code{|>}, and downstream tools such as \pkg{dplyr} \citep{R-dplyr}, \pkg{purrr} \citep{R-purrr}, \pkg{ggplot2} \citep{R-ggplot2}, and presentation-table packages such as \pkg{gt} \citep{R-gt} and \pkg{kableExtra} \citep{R-kableExtra}.
The contrast with the heterogeneous return types of \code{stats::cor.test()} (an \code{htest} object), \code{polycor::polychor()} (a numeric scalar or a custom \code{polycor} list), and \code{stats::chisq.test()} (another \code{htest}) is what motivated the unified API: practitioners no longer have to write per-method wrangling code to assemble a comparable summary across mixed variable types.

The \code{smart\_cor\_df()} function is a non-interactive convenience wrapper that returns a plain list of two data frames (correlations and methods), with latent normality assessed per pair under the default \code{assume\_latent\_normal = "auto"}.
It is designed for use in automated pipelines where interactive prompts are undesirable.

\begin{CodeInput}
R> library("dplyr")
R> mat <- smart_cormat(data, verbose = FALSE, assume_latent_normal = TRUE)
R> tidy(mat) |>
+    filter(abs(estimate) > 0.3) |>
+    arrange(desc(abs(estimate)))
\end{CodeInput}

\noindent

The tidied data frame carries the full inference machinery for each pair, including \code{ci\_lower}, \code{ci\_upper}, \code{ci\_method}, \code{ci\_source}, \code{p\_method}, \code{null\_hypothesis}, and \code{rationale}, so downstream code can filter or annotate by source as easily as by estimate:

\begin{CodeInput}
R> mtcars |>
+    smart_cormat(verbose = FALSE) |>
+    tidy() |>
+    filter(ci_method == "analytic", abs(estimate) > 0.5) |>
+    select(var_x, var_y, method, estimate, ci_lower, ci_upper, ci_source)
\end{CodeInput}

\noindent
Group-wise correlations follow the same idiomatic approach via \pkg{dplyr}'s \code{group\_modify()}:

\begin{CodeInput}
R> mtcars |>
+    group_by(am) |>
+    group_modify(~ tidy(smart_cormat(.x, verbose = FALSE))) |>
+    ungroup()
\end{CodeInput}

\noindent
\pkg{purrr}'s \code{map\_dfr()} maps over arbitrary variable pairs while preserving the consistent return shape:

\begin{CodeInput}
R> pairs <- list(c("mpg","wt"), c("mpg","hp"), c("disp","qsec"))
R> purrr::map_dfr(pairs, ~ tidy(smart_cor(mtcars[[.x[1]]], mtcars[[.x[2]]],
+                                            verbose = FALSE)))
\end{CodeInput}

\noindent
For publication-quality output, \pkg{gt} or \pkg{kableExtra} can format the tidied tibble directly: \code{tidy(mat) |> gt::gt() |> gt::fmt\_number(decimals = 3)}.

\noindent
Note that \code{smart\_cor\_df()} fixes \code{verbose = FALSE} internally; users requiring verbose output should call \code{smart\_cormat()} directly.

\section{Case study: General Social Survey} \label{sec:casestudy}

To demonstrate \pkg{smartcor} on real data, we analyze a subset of the General Social Survey \citep[GSS;][]{GSS2024}, accessed via the \pkg{gssr} \proglang{R} package \citep{R-gssr}.
The GSS is a long-running, mostly annual survey of United States households administered by NORC, and it pairs continuous, ordinal, binary, and count items in a single instrument, which makes it a natural test bed for a package whose purpose is to match the correlation measure to the variable type.
Because \pkg{gssr} bundles the cumulative file and exposes the codebook through \proglang{R}'s help system, the variables and their response scales are inspectable without leaving the analysis session.\footnote{The package is maintained by Kieran Healy and hosted at \url{https://github.com/kjhealy/gssr}, with documentation at \url{https://kjhealy.github.io/gssr/}.}

\subsection{Data and variables} \label{sec:gss-data}

We select seven variables representing all four types:

\begin{itemize}
  \item \textbf{Continuous:} \code{age} (age in years), \code{coninc}
    (family income in constant dollars).
  \item \textbf{Ordinal:} \code{degree} (highest degree: less than high
    school, high school, junior college, bachelor, graduate),
    \code{happy} (general happiness, coded $1$ = very happy, $2$ = pretty happy, $3$ = not too happy, so negative correlations with \code{happy} indicate greater happiness).
  \item \textbf{Binary:} \code{sex} (male/female), \code{marital}
    (currently married vs. not, recoded from 5 to 2 unique categories by us).
  \item \textbf{Categorical:} \code{region} (4 Census regions: Northeast, Midwest, South, West).
\end{itemize}

The GSS distributes data in Stata format with \code{haven\_labelled} vectors, which inherit from \code{double}.
A researcher who loads the data and immediately calls \code{cor()} encounters no error:

\begin{CodeInput}
R> library("gssr")
R> library("dplyr")
R> gss <- gss_get_yr(2024)
R> gss_raw <- gss[, c("age", "coninc", "degree", "happy",
+                      "sex", "marital", "region")]
R> round(cor(gss_raw, use = "pairwise.complete.obs"), 3)
\end{CodeInput}
\begin{CodeOutput}
           age coninc degree  happy    sex marital region
age      1.000  0.023  0.051 -0.058  0.007  -0.380 -0.061
coninc   0.023  1.000  0.436 -0.150 -0.107  -0.330  0.055
degree   0.051  0.436  1.000 -0.087  0.004  -0.152  0.008
happy   -0.058 -0.150 -0.087  1.000 -0.017   0.216  0.008
sex      0.007 -0.107  0.004 -0.017  1.000   0.027 -0.033
marital -0.380 -0.330 -0.152  0.216  0.027   1.000 -0.009
region  -0.061  0.055  0.008  0.008 -0.033  -0.009  1.000
\end{CodeOutput}

\noindent
The silence is dangerous.
Every value in this matrix is a Pearson correlation computed on raw numeric codes.
The \code{region} column (where the integer codes $1$--$4$ label the four Census regions $1$ = Northeast, $2$ = Midwest, $3$ = South, $4$ = West) produces a small but meaningless correlation with income, an artifact of the arbitrary ordering imposed on an unordered category.
The \code{marital} column (1 = married, 2 = widowed, \ldots, 5 = never married) appears to ``correlate'' $-0.38$ with age, but this number reflects nothing about the actual association between age and marital status.
A careful analyst would convert \code{marital} from categorical with five categories to a binary variable (married/not married), as we share above and demonstrate below.
Similarly, \code{sex} is also converted to a 0/1 binary variable from original 1/2 coding, where $2 := \text{Female}$.
The \code{happy} and \code{degree} columns use ordinal codes as if they were continuous measurements.

A more careful analyst converts each variable to its appropriate \proglang{R} type before computing correlations:

\begin{CodeInput}
R> set.seed(2026)
R> gss_sub <- gss |>
+    transmute(
+      age     = as.numeric(age),
+      coninc  = as.numeric(coninc),
+      degree  = factor(haven::as_factor(degree), ordered = TRUE),
+      happy   = factor(haven::as_factor(happy), ordered = TRUE),
+      sex     = as.integer(as.numeric(sex) == 2),
+      marital = as.integer(as.numeric(marital) == 1),
+      region  = haven::as_factor(region)
+    ) |>
+    filter(complete.cases(age, coninc, degree, happy,
+                          sex, marital, region)) |>
+    slice(sample.int(n(), 3000))
\end{CodeInput}

\noindent
To make every result below reproducible identically in \proglang{R} and \proglang{Python}, we freeze this $3{,}000$-row complete-case extract as \code{gss\_2024\_casestudy.csv}, which both reproduction bundles read.

\noindent
But now \code{cor()} fails entirely:

\begin{CodeInput}
R> cor(gss_sub, use = "pairwise.complete.obs")
\end{CodeInput}
\begin{CodeOutput}
Error in cor(gss_sub, use = "pairwise.complete.obs") :
  'x' must be numeric
\end{CodeOutput}

\subsection{The naive workaround} \label{sec:naive}

\noindent
The typical workaround is to drop non-numeric columns and compute Pearson on whatever remains:

\begin{CodeInput}
R> gss_num <- gss_sub[, sapply(gss_sub, is.numeric)]
R> round(cor(gss_num, use = "pairwise.complete.obs"), 3)
\end{CodeInput}
\begin{CodeOutput}
           age coninc    sex marital
age      1.000  0.019 -0.009   0.141
coninc   0.019  1.000 -0.101   0.393
sex     -0.009 -0.101  1.000  -0.090
marital  0.141  0.393 -0.090   1.000
\end{CodeOutput}

\noindent
Three of the seven variables (\code{degree}, \code{happy}, \code{region}) are silently dropped.
The surviving four columns all receive Pearson correlations, which underestimate associations involving binary variables (\autoref{sec:equiv}) and ignore ordinal structure entirely.

\subsection[The smartcor approach]{The \pkg{smartcor} approach} \label{sec:smart-approach}

\begin{CodeInput}
R> mat <- smart_cormat(gss_sub, verbose = FALSE)
R> mat
\end{CodeInput}
\begin{CodeOutput}
-- Smart Correlation Matrix --

-- Variable types --

age: "count"
coninc: "continuous"
degree: "ordinal"
happy: "ordinal"
sex: "binary"
marital: "binary"
region: "categorical"

-- Correlations --

            age coninc degree   happy     sex marital region
age      1.0000  0.019  0.028 -0.0680 -0.0087   0.140  0.068
coninc   0.0190  1.000  0.450 -0.1700 -0.1000   0.390  0.092
degree   0.0280  0.450  1.000 -0.0790  0.0150   0.180  0.058
happy   -0.0680 -0.170 -0.079  1.0000 -0.0047  -0.230  0.035
sex     -0.0087 -0.100  0.015 -0.0047  1.0000  -0.090  0.063
marital  0.1400  0.390  0.180 -0.2300 -0.0900   1.000  0.067
region   0.0680  0.092  0.058  0.0350  0.0630   0.067  1.000
-- Methods used --

        age coninc degree happy sex marital region
age     Pr  Pr     Sp     Ps    PB  PB      CV
coninc  Pr  Pr     Sp     Sp    PB  PB      CV
degree  Sp  Sp     Pc     Kn    RB  RB      CV
happy   Ps  Sp     Kn     Pc    RB  RB      CV
sex     PB  PB     RB     RB    Ph  Ph      CV
marital PB  PB     RB     RB    Ph  Ph      CV
region  CV  CV     CV     CV    CV  CV      CV

i Legend:
"Pr" = Pearson Correlation
"Sp" = Spearman Rank Correlation
"Ps" = Polyserial Correlation
"PB" = Point-Biserial Correlation (= Pearson)
"CV" = Cramer's V
"Pc" = Polychoric Correlation
"Kn" = Kendall's Tau-b
"RB" = Rank-Biserial Correlation
"Ph" = Phi Coefficient (= Pearson for 0/1)
\end{CodeOutput}

\noindent
\code{age} is recognised as a count (integer years 18--89, 70+ unique values), and counts are routed through the same Pearson default as continuous variables, so \code{age}--\code{coninc} returns \code{Pr}. Every diagonal cell records the method that would have been chosen if that variable had been correlated with another variable of the same type: \code{Pr} for continuous-continuous and count-count pairs, \code{Pc} for ordinal--ordinal pairs, \code{Ph} for binary--binary pairs (under the default \code{assume\_latent\_normal = "auto"} a saturated $2 \times 2$ table resolves to $\phi$; tetrachoric requires \code{assume\_latent\_normal = TRUE}), and \code{CV} for categorical. The numeric value at every diagonal is $1.000$ by definition.

\noindent
The correlation matrix \code{mat} includes all seven variables. The methods matrix reveals the analytic heterogeneity: Pearson for continuous--continuous pairs, polychoric or Kendall's~$\tau$ for ordinal--ordinal pairs (the choice depends on the data-driven latent-normality test, which is rejected for \code{degree}--\code{happy} so \code{Kn} is selected, but the diagonals retain \code{Pc} as the conceptual default), polyserial or Spearman for continuous--ordinal and count--ordinal pairs (\code{age}--\code{degree} test rejects, hence \code{Sp}; \code{age}--\code{happy} accepts, hence \code{Ps}), point-biserial for continuous--binary pairs, $\phi$ for the binary--binary pair, and rank-biserial for binary--ordinal pairs.
The \code{region} column reports Cram\'{e}r's~$V$ values, which measure statistical association rather than directional correlation. These values are on a different scale ($[0, 1]$ with no sign) and should not be compared directly to correlation coefficients (see \autoref{sec:when}). The methods matrix identifies \code{CV} for these pairs, alerting the analyst to this distinction.

The heatmaps (Figures~\ref{fig:cor-heatmap} and~\ref{fig:method-heatmap}) provide a visual summary:

\begin{CodeInput}
R> ggcor_heatmap(mat)
R> ggcor_method_heatmap(mat)
\end{CodeInput}

\begin{figure}[t]
  \centering
  \includegraphics[width=0.8\textwidth]{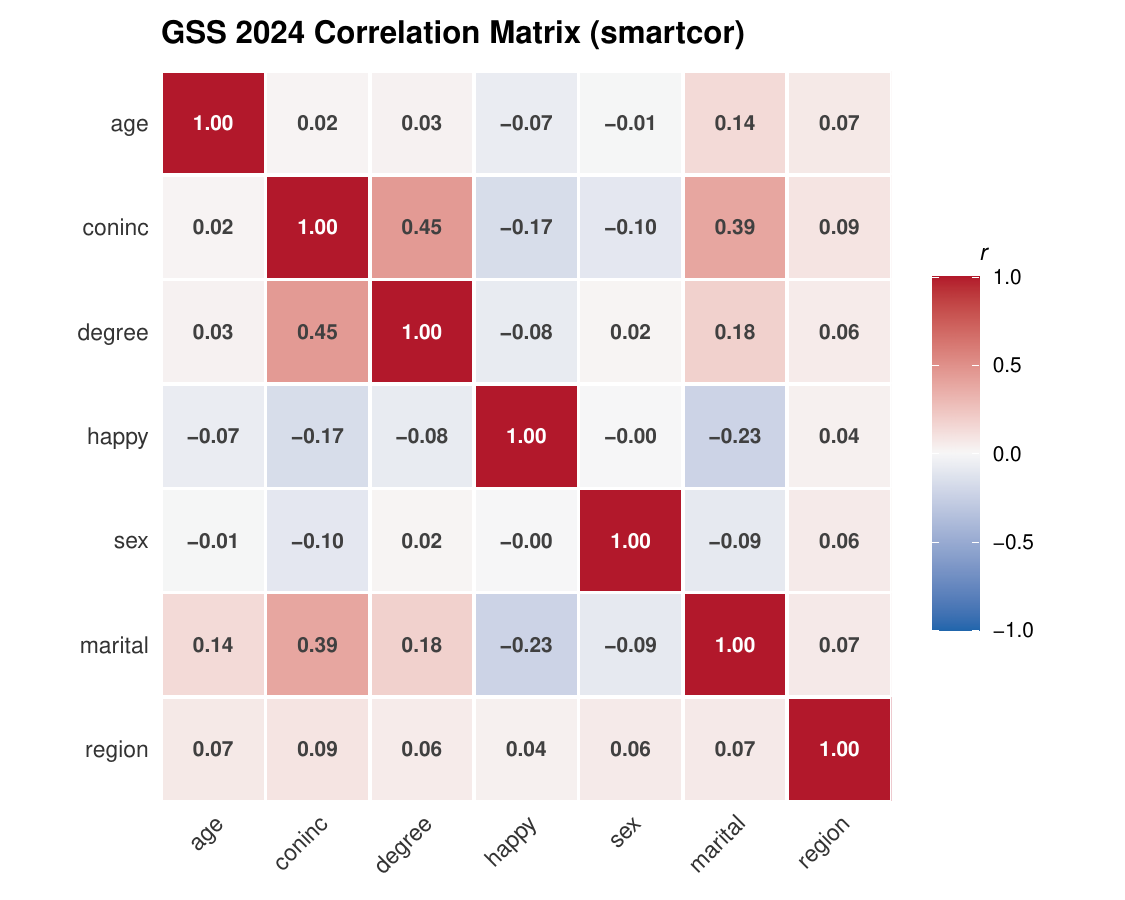}
  \caption{Correlation heatmap for the GSS subset using \code{smart\_cormat()}.
    Each cell displays the estimated correlation, with the color scale ranging
    from $-1$ (blue) to $+1$ (red). Different methods are used for each pair
    depending on variable types.}
  \label{fig:cor-heatmap}
\end{figure}

\begin{figure}[t]
  \centering
  \includegraphics[width=0.8\textwidth]{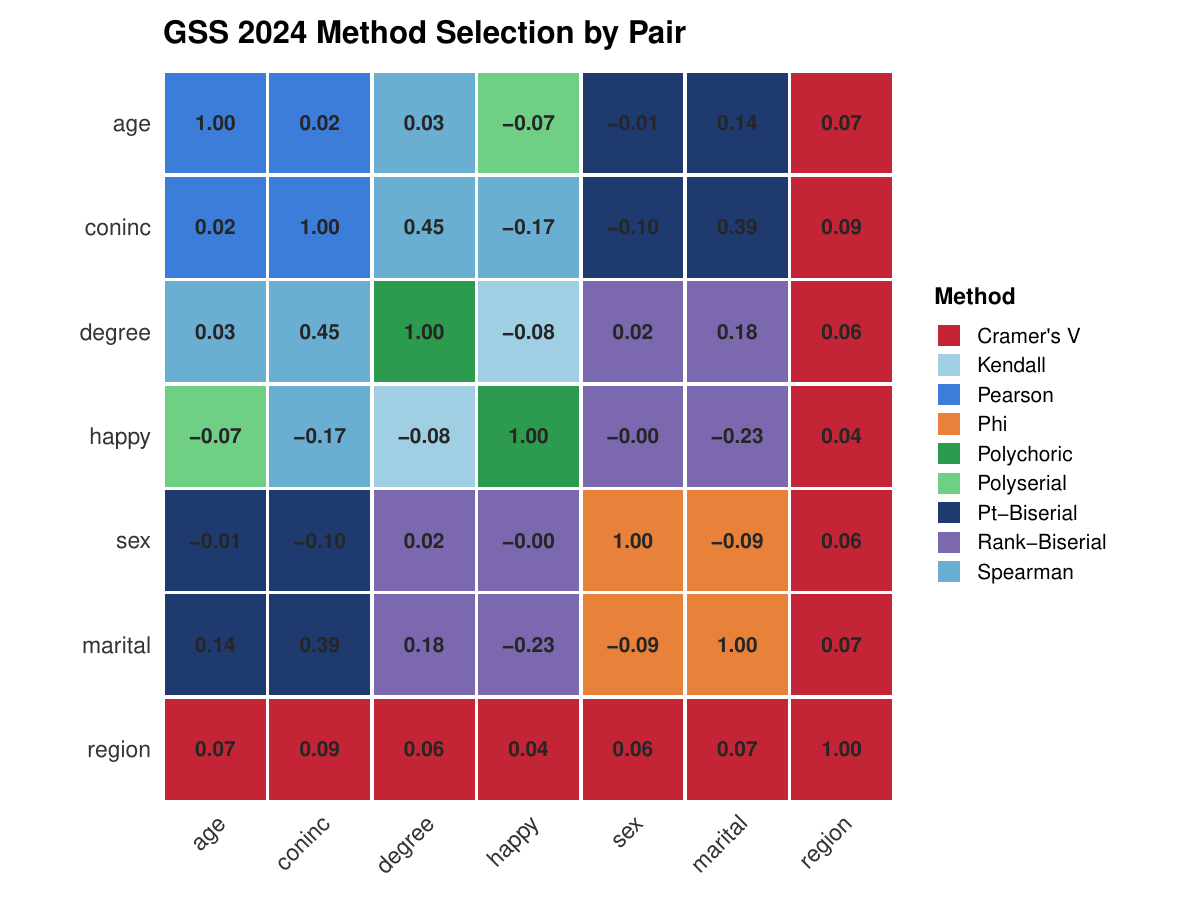}
  \caption{Method-selection heatmap showing which correlation method was
    selected for each variable pair. The analytic heterogeneity is immediately
    visible: nine distinct methods appear in a single correlation matrix.}
  \label{fig:method-heatmap}
\end{figure}

\subsection{Comparing methods for a specific pair} \label{sec:gss-compare}

To illustrate the impact of method choice, we compare all applicable methods for the \code{degree}--\code{happy} pair (ordinal--ordinal):

\begin{CodeInput}
R> compare_methods(gss_sub$degree, gss_sub$happy, verbose = FALSE)
\end{CodeInput}
\begin{CodeOutput}
-- Method Comparison --
Variables: gss_sub$degree ("ordinal") x gss_sub$happy ("ordinal")
N: 3000

Polychoric Correlation
r = -0.1031 p < 0.001 95% CI [-0.1456, -0.0601]
Kendall's Tau-b <- recommended
r = -0.0795 p < 0.001 95% CI [-0.1030, -0.0559]
Spearman Rank Correlation
r = -0.0906 p < 0.001 95% CI [-0.1261, -0.0549]
Goodman-Kruskal's Gamma
r = -0.1261 p < 0.001 95% CI [-0.1746, -0.0769]
\end{CodeOutput}

\noindent
All four methods report analytic 95\% confidence intervals and $p$-values, drawn from the sources catalogued in \autoref{tab:inference-sources}: Kendall (Fieller--Hartley--Pearson), Spearman (Bonett--Wright), polychoric (Olsson/J\"oreskog, Fisher~$z$), and Goodman--Kruskal~$\gamma$ (Brown--Benedetti). No bootstrap is invoked under the default \code{bootstrap = "auto"} for this pair.
The polychoric estimate ($-0.103$) exceeds the Spearman ($-0.091$) and Kendall ($-0.079$) estimates in absolute value, consistent with the attenuation theory demonstrated in \autoref{sec:attenuation}. 
But since the latent-normality test is rejected for this pair, \code{compare\_methods()} recommends Kendall's~$\tau$, the value \code{smart\_cor()} would report by default.
Gamma ($-0.126$) is the largest in absolute value but also has the widest interval ($[-0.175, -0.077]$), because it ignores tied pairs entirely and the resulting effective sample size is smaller.
This comparison enables the researcher to make an informed choice about which estimate to report; \code{smart\_cor()} provides a transparent default while \code{compare\_methods()} reveals the sensitivity.

\subsection{Pipeline integration} \label{sec:gss-pipeline}

The tidy output integrates smoothly with \pkg{dplyr} pipelines:

\begin{CodeInput}
R> tidy(mat) |>
+    filter(x_type == "ordinal" | y_type == "ordinal") |>
+    select(var_x, var_y, method, estimate) |>
+    arrange(desc(abs(estimate)))
\end{CodeInput}
\begin{CodeOutput}
# A tibble: 11 x 4
   var_x  var_y   method        estimate
   <chr>  <chr>   <chr>            <dbl>
 1 coninc degree  spearman        0.451
 2 happy  marital rank_biserial  -0.230
 3 degree marital rank_biserial   0.185
 4 coninc happy   spearman       -0.168
 5 degree happy   kendall        -0.0795
 6 age    happy   polyserial     -0.0676
 7 degree region  cramers_v       0.0580
 8 happy  region  cramers_v       0.0351
 9 age    degree  spearman        0.0276
10 degree sex     rank_biserial   0.0154
11 happy  sex     rank_biserial  -0.00474
\end{CodeOutput}

\noindent
This pipeline extracts all pairs involving at least one ordinal variable and sorts them by correlation strength, allowing the researcher to quickly identify the strongest associations.

\subsection{Key findings} \label{sec:gss-findings}

The comparison of naive Pearson versus type-aware \pkg{smartcor} analysis reveals three patterns.
\textit{First}, the latent-variable estimators recover more association than Pearson where they apply: \code{compare\_methods()} reports polychoric $-0.103$ versus Pearson $-0.091$ for \code{degree}--\code{happy}, consistent with the attenuation theory in \autoref{sec:attenuation}; for this particular pair, however, the latent-normality test rejects, so \code{smart\_cor()} conservatively reports Kendall's~$\tau$ instead.
\textit{Second}, the categorical variable \code{region} now has association measures (Cram\'{e}r's~$V$) with all other variables, whereas the naive approach excluded it entirely. 
These are measures of statistical association rather than correlation proper (see \autoref{sec:when}), but they provide useful information about dependence that would otherwise be lost.
\textit{Third}, the methods matrix makes the analytic heterogeneity visible: eight different methods are used across the variable pairs of a single correlation matrix (nine counting the diagonal's self-pair method), each chosen for a principled reason.

\subsection{Overriding type detection} \label{sec:gss-override}

The automatic type detection in \code{smart\_cor()} can be overridden for individual variables using the \code{x\_type} and \code{y\_type} arguments.
This is useful when the heuristic classification does not match domain knowledge.
For example, \code{mtcars\$cyl} has only three unique values (4, 6, 8) and is automatically detected as ordinal, but a researcher might prefer to treat it as continuous:

\begin{CodeInput}
R> # Default: cyl detected as ordinal -> Spearman
R> smart_cor(mtcars$mpg, mtcars$cyl,
+    assume_latent_normal = FALSE, verbose = FALSE)
\end{CodeInput}
\begin{CodeOutput}
-- Smart Correlation --
Estimate: -0.9108
Method: Spearman Rank Correlation
Variables: mtcars$mpg ("continuous") x mtcars$cyl ("ordinal")
N: 32
p-value: < 0.001
Test: t approximation (cor.test, exact = FALSE)
H0: rho_S = 0
Small p-values (e.g., p < 0.05) indicate evidence against H0.
95% CI: [-0.9615, -0.8002]
Source: Bonett and Wright (2000)

One variable is continuous and the other is ordinal. Spearman rank correlation
selected (no latent normality assumption).

i Alternatives: "kendall" and "polyserial" (pass `method = "..."` to use)
\end{CodeOutput}

\begin{CodeInput}
R> # Override: treat cyl as continuous -> Pearson
R> smart_cor(mtcars$mpg, mtcars$cyl,
+    y_type = "continuous", verbose = FALSE)
\end{CodeInput}
\begin{CodeOutput}
-- Smart Correlation --
Estimate: -0.8522
Method: Pearson Correlation
Variables: mtcars$mpg ("continuous") x mtcars$cyl ("continuous")
N: 32
p-value: < 0.001
Test: t-test on r (cor.test)
H0: rho = 0
Small p-values (e.g., p < 0.05) indicate evidence against H0.
95% CI: [-0.9258, -0.7163]
Source: Fisher z (cor.test)

Both variables are continuous; Pearson correlation selected.

i Alternatives: "spearman" and "kendall" (pass `method = "..."` to use)
\end{CodeOutput}

\noindent
The Spearman estimate ($-0.911$) exceeds the Pearson ($-0.852$) because Spearman captures the strong monotonic relationship without assuming linearity.
The \code{x\_type} argument overrides only the specified variable; the other retains its auto-detected type.
This design ensures that type overrides are explicit and local, rather than requiring a global reconfiguration.

\section{Comparison with existing software} \label{sec:comparison}

Several \proglang{R} packages compute correlations for mixed variable types, but none covers all 10 type-pair combinations or explains its method choices.
Table~\ref{tab:feature-comparison} compares \pkg{smartcor} with three widely used alternatives.

\begin{table}[t]
\centering
\caption{Feature comparison of \proglang{R} packages for mixed-type correlation analysis.}
\label{tab:feature-comparison}
\begin{tabular}{lcccc}
\toprule
Feature & \pkg{smartcor} & \pkg{correlation} & \pkg{polycor} & \pkg{psych} \\
\midrule
Auto type detection & 5 types & Partial & Limited & Threshold \\
All 10 type-pair combos & Yes & No & No & No \\
Categorical measures & Yes & No & No & No \\
Latent normality control & Yes & No & No & No \\
Explains method choice & Yes & No & No & No \\
Tidy output & Yes & Yes & No & No \\
\code{compare\_methods()} & Yes & No & No & No \\
\proglang{Python} companion & Yes & No & No & No \\
\bottomrule
\end{tabular}
\end{table}

\subsection{Code comparison} \label{sec:code-comparison}

Consider computing a correlation matrix for a data frame with continuous, binary, ordinal, and categorical variables.
With \pkg{smartcor}, this requires a single function call:

\begin{CodeInput}
R> smart_cormat(data, assume_latent_normal = TRUE, verbose = FALSE)
\end{CodeInput}

\noindent
With \pkg{polycor}, the user must manually specify which variables are ordinal and binary, and only polychoric/polyserial correlations are computed:

\begin{CodeInput}
R> library("polycor")
R> hetcor(data, ML = TRUE)
\end{CodeInput}

\noindent
The \code{hetcor()} function from \pkg{polycor} handles Pearson, polyserial, and polychoric correlations but ignores categorical variables entirely: it computes neither Cram\'{e}r's~$V$ nor Theil's~$U$, and it provides no rationale for its choices.
The \code{psych::mixedCor()} function uses a numeric threshold to distinguish continuous from categorical variables but does not differentiate binary from ordinal, limiting its coverage of type-pair combinations.

The \pkg{correlation} package \citep{R-correlation} provides a \code{method = "auto"} option, but its auto-detection does not reliably distinguish ordinal from continuous variables, and it does not support categorical association measures.

\subsection{Key differentiators} \label{sec:differentiators}

\pkg{smartcor} is a \emph{decision-support tool}, not merely a method library.
No other package provides: (1)~a human-readable rationale for each method choice; (2)~coverage of all 10 variable-type pair combinations including categorical measures; (3)~a \code{compare\_methods()} function for sensitivity analysis; or (4)~selection logic grounded in simulation evidence (\autoref{sec:when}).

\section[The pysmartcor package for Python]{The \pkg{pysmartcor} package for \proglang{Python}} \label{sec:python}

A companion \proglang{Python} package provides the same decision logic and API.
The package uses \pkg{NumPy} \citep{numpy}, \pkg{SciPy} \citep{scipy}, and \pkg{pandas} \citep{pandas} for core computation; polychoric and polyserial correlations are estimated by a vendored, corrected port of \pkg{semopy}'s maximum-likelihood routines (the role \pkg{polycor} plays in the \proglang{R} implementation), so these three libraries are the only runtime dependencies.\footnote{Vendoring became necessary for two reasons.
First, \pkg{semopy}~2.3.11 evaluates bivariate-normal rectangle probabilities through \code{scipy.stats.mvn}, a private interface removed in SciPy~1.14, so its polychoric and polyserial routines raise \code{AttributeError} on any recent SciPy; the port replaces those calls with \code{scipy.stats.multivariate\_normal.cdf} via inclusion--exclusion.
Second, \pkg{semopy}'s polyserial likelihood standardizes the continuous variable as $z = (x - \bar{x})/s^2$ --- dividing by the variance where the standard deviation is required --- and additionally passes the variance to the \code{scale} argument of \code{norm.logpdf}, which expects a standard deviation; both errors bias $\hat\rho$.
The vendored port fixes both and is distributed under \pkg{semopy}'s MIT license with attribution.}

\subsection{Installation}

\pkg{pysmartcor}~1.0.0 is published on the Python Package Index (PyPI) at \url{https://pypi.org/project/pysmartcor/} and installs with \code{pip}:

\begin{CodeInput}
$ pip install pysmartcor
$ pip install pysmartcor[viz]   # includes matplotlib, seaborn
\end{CodeInput}

\subsection{API parity}

The \proglang{Python} API mirrors the \proglang{R} API:

\begin{CodeInput}
>>> from pysmartcor import smart_cor, smart_cormat, compare_methods
>>> result = smart_cor(x, y, assume_latent_normal=True, verbose=False)
>>> round(result.estimate, 4)
-0.8677
>>> result.to_dataframe()   # equivalent to tidy() in R
\end{CodeInput}

\noindent
The same functions are available: \code{smart\_cor()}, \code{smart\_cormat()}, \code{smart\_cor\_df()}, \code{compare\_methods()}, \code{detect\_type()}, and \code{available\_methods()}.

Table~\ref{tab:api-comparison} shows the correspondence between key elements of the \proglang{R} and \proglang{Python} APIs.

\begin{table}[h]
\centering
\caption{API correspondence between the \proglang{R} and \proglang{Python} packages.}
\label{tab:api-comparison}
\begin{tabular}{ll}
\toprule
\proglang{R} & \proglang{Python} \\
\midrule
\code{result\$estimate} & \code{result.estimate} \\
\code{tidy(result)} & \code{result.to\_dataframe()} \\
\code{tidy(mat)} & \code{mat.to\_long()} \\
\code{mat\$correlations} & \code{mat.correlations} \\
\code{ggcor\_heatmap(mat)} & \code{cor\_heatmap(mat)} \\
\code{smart\_cor\_df(data)} & \code{smart\_cor\_df(data)} \\
\bottomrule
\end{tabular}
\end{table}

\subsection{Visualization}

The optional \code{viz} module provides \pkg{matplotlib} \citep{matplotlib} and \pkg{seaborn} \citep{seaborn} based equivalents of the \proglang{R} heatmap functions:

\begin{CodeInput}
>>> from pysmartcor.viz import cor_heatmap, method_heatmap
>>> fig = cor_heatmap(mat)
>>> fig = method_heatmap(mat)
\end{CodeInput}

\section{Concluding remarks} \label{sec:summary}

Pearson correlation remains the default in virtually every statistical software environment, yet it is only one member of a large family of association measures, each designed for specific variable types and distributional assumptions.
When applied indiscriminately to binary, ordinal, or mixed-type data, Pearson's $r$ can attenuate, inflate, or otherwise distort the strength of association.
The \pkg{smartcor} package addresses this gap by automatically detecting variable types, selecting the statistically appropriate method from a library of 14~correlation measures, and explaining the rationale behind each choice.
The Monte~Carlo simulations presented in \autoref{sec:when} quantified when method selection matters and when it does not, while the General Social Survey case study in \autoref{sec:casestudy} demonstrated that type-aware correlation analysis yields substantively different results on real data.
Both the \proglang{R} and \proglang{Python} implementations share identical decision logic, ensuring reproducibility across platforms.
Both implementations are distributed through the standard repositories: \pkg{smartcor} via CRAN and \pkg{pysmartcor} via PyPI (\url{https://pypi.org/project/pysmartcor/}).

The simulation evidence supports several practical conclusions.
Pearson correlation is numerically identical to the point-biserial coefficient for continuous-binary pairs and to the $\phi$ coefficient for binary-binary pairs, provided the binary variables are coded 0/1; practitioners gain nothing by switching methods in these cases.
The picture changes for ordinal data.
Discretizing a continuous variable into ordered categories attenuates the observed correlation by roughly 2--21\%, with fewer categories producing greater underestimation (dichotomization is worse still, attenuating by up to ${\approx}37\%$).
Polychoric and polyserial correlations recover the latent association when the underlying bivariate normality assumption holds, but they exhibit substantial bias when it does not.
Among rank-based alternatives, Kendall's $\tau_b$ proved the most robust estimator across non-normal latent distributions, making it the safer default when distributional assumptions are uncertain.

The current implementation has clear boundaries.
The \code{assume\_latent\_normal} argument gives users explicit control over the latent normality assumption: \code{TRUE} and \code{FALSE} force the choice, \code{NULL} asks interactively, and \code{"auto"} runs the likelihood-ratio goodness-of-fit test described in \autoref{sec:normality} to decide from the data.
The \code{"auto"} option implements the model-fit test proposed by \citet{joreskog2005structural}, selecting polychoric correlation when the bivariate normal threshold model fits the observed cell frequencies and falling back to Kendall's~$\tau$ when it does not.
All correlations are bivariate: controlling for confounders through partial correlation requires separate analysis.
The inferential framework is entirely frequentist.
Every supported method except Theil's~$U$ ships with an analytic confidence interval by default; the closed-form / asymptotic sources are catalogued in \autoref{tab:inference-sources}.
Bootstrap confidence intervals remain available via the \code{bootstrap} argument as either an opt-in (\code{bootstrap = TRUE}) for cross-checking analytic intervals, or as automatic fallback when an analytic interval fails (e.g., singular Hessian, sparse $2 \times 2$ table); the analytic or permutation $p$-value is retained.
When \code{bootstrap = "auto"} (the default), the percentile bootstrap is invoked only when no analytic CI is available for the chosen method.
The \code{conf\_level} parameter controls the confidence level (default: $0.95$), and \code{n\_boot} sets the number of bootstrap replications (default: $500$, lowered from $1{,}000$ now that bootstrap is a fallback path rather than the primary one).
These boundaries reflect deliberate design choices rather than oversights; each represents a direction for future development.

\section*{Acknowledgments}

The \pkg{smartcor} package relies on the foundational \proglang{R} package \pkg{polycor} for polychoric, polyserial, and tetrachoric correlation estimation.
We thank its maintainers for their sustained contributions to the \proglang{R} ecosystem.

\bibliographystyle{abbrvnat}
\bibliography{smartcor}

\newpage
\appendix

\section{Proofs of correlation equivalences} \label{app:proofs}

\subsection{Equality of Pearson and point-biserial correlations}
\label{app:pearson-biserial}

\begin{assumption}\label{ass:biserial}
The data $\{(x_i, y_i)\}_{i=1}^n$ satisfy:
(i) $n \geq 2$;
(ii) $y_i \in \{0,1\}$ with both groups non-empty, so $n_1 = \sum_i y_i$ and $n_0 = n - n_1$ both lie in $\{1, \ldots, n-1\}$, equivalently $\hat{p} = n_1/n \in (0,1)$;
(iii) $x_i \in \mathbb{R}$ has positive sample variance, $s_x^2 > 0$.
\end{assumption}

\begin{theorem}[Pearson equals point-biserial]\label{thm:pearson-biserial}
Under Assumption~\ref{ass:biserial}, the Pearson sample correlation between $x$ and $y$ coincides with the point-biserial coefficient,
\[
r \;=\; \frac{\bar{x}_1 - \bar{x}_0}{s_x}\sqrt{\hat{p}(1-\hat{p})} \;=\; r_{pb},
\]
where $\bar{x}_j = n_j^{-1}\sum_{y_i=j} x_i$ for $j \in \{0,1\}$.
\end{theorem}

\begin{proof}
The overall mean decomposes as $\bar{x} = \hat{p}\bar{x}_1 + (1-\hat{p})\bar{x}_0$. Because $y_i \in \{0,1\}$, the deviation $y_i - \hat{p}$ takes only the two values $1-\hat{p}$ and $-\hat{p}$, so the sample covariance becomes
\begin{align*}
S_{xy} &= \frac{1}{n-1}\sum_{i=1}^n (x_i-\bar{x})(y_i-\hat{p})\\
&= \frac{1}{n-1}\left[n_1(\bar{x}_1-\bar{x})(1-\hat{p}) + n_0(\bar{x}_0-\bar{x})(-\hat{p})\right].
\end{align*}
Substituting $\bar{x}_1 - \bar{x} = (1-\hat{p})(\bar{x}_1 - \bar{x}_0)$ and $\bar{x}_0 - \bar{x} = -\hat{p}(\bar{x}_1 - \bar{x}_0)$ collapses both terms,
\[
S_{xy} = \frac{n}{n-1}\,\hat{p}(1-\hat{p})(\bar{x}_1 - \bar{x}_0).
\]
For the binary variable, $s_y^2 = \frac{n}{n-1}\,\hat{p}(1-\hat{p})$, strictly positive by Assumption~\ref{ass:biserial}(ii). Combined with $s_x > 0$ from Assumption~\ref{ass:biserial}(iii),
\[
r = \frac{S_{xy}}{s_x s_y} = \frac{\bar{x}_1 - \bar{x}_0}{s_x}\sqrt{\hat{p}(1-\hat{p})} = r_{pb}. \qedhere
\]
\end{proof}

\begin{remark}[Population analogue]
With $Y \in \{0,1\}$, $p = \Pr(Y=1) \in (0,1)$, $\mu_j = \mathbb{E}[X|Y=j]$, and $\sigma_X > 0$, the same argument yields $\rho = (\mu_1 - \mu_0)\sigma_X^{-1}\sqrt{p(1-p)}$.
\end{remark}

\begin{remark}[On the coding of $y$]
Pearson's $r$ is invariant under any positive affine recoding $y \mapsto \gamma y + \delta$ with $\gamma > 0$, so the numerical correlation is unchanged if $y$ is coded as $\{-1, +1\}$ rather than $\{0,1\}$. The closed form $r_{pb} = (\bar{x}_1 - \bar{x}_0)s_x^{-1}\sqrt{\hat{p}(1-\hat{p})}$ is, however, derived under $0/1$ coding and takes different algebraic shapes under other codings, even though the numerical value of the correlation it represents is the same.
\end{remark}

\subsection{Equality of Pearson and $\phi$ for two binary variables}
\label{app:pearson-phi}

\begin{assumption}\label{ass:phi}
The data $\{(x_i, y_i)\}_{i=1}^n$ satisfy:
(i) $n \geq 2$;
(ii) $x_i, y_i \in \{0,1\}$;
(iii) both marginals are non-degenerate, i.e., neither row sum nor column sum of the $2\times 2$ contingency table equals $0$ or $n$.
\end{assumption}

\begin{theorem}[Pearson equals $\phi$]\label{thm:pearson-phi}
Under Assumption~\ref{ass:phi}, the Pearson sample correlation between $x$ and $y$ equals the $\phi$ coefficient,
\[
r \;=\; \frac{N_{11}N_{00} - N_{10}N_{01}}{\sqrt{N_{1\cdot}\,N_{0\cdot}\,N_{\cdot 1}\,N_{\cdot 0}}} \;=\; \phi,
\]
where $N_{ij}$ is the count of observations with $x = i$ and $y = j$, and $N_{i\cdot}$, $N_{\cdot j}$ denote row and column marginals respectively.
\end{theorem}

\begin{proof}
Arrange the contingency counts as
\[
\begin{array}{c|cc|c}
 & y=1 & y=0 & \text{row sums}\\\hline
x=1 & N_{11} & N_{10} & N_{1\cdot}\\
x=0 & N_{01} & N_{00} & N_{0\cdot}\\\hline
\text{col sums} & N_{\cdot 1} & N_{\cdot 0} & n
\end{array}
\]
with $n = N_{11} + N_{10} + N_{01} + N_{00}$, and write $\bar{x} = N_{1\cdot}/n$, $\bar{y} = N_{\cdot 1}/n$. Because indicators satisfy $x_i^2 = x_i$ and $y_i^2 = y_i$,
\begin{align*}
s_x^2 &= \frac{n}{n-1}\,\bar{x}(1-\bar{x}) = \frac{N_{1\cdot}\,N_{0\cdot}}{n(n-1)},\\
s_y^2 &= \frac{n}{n-1}\,\bar{y}(1-\bar{y}) = \frac{N_{\cdot 1}\,N_{\cdot 0}}{n(n-1)},
\end{align*}
both strictly positive by Assumption~\ref{ass:phi}(iii). For the covariance,
\begin{align*}
S_{xy}
&= \frac{1}{n-1}\sum_{i=1}^n (x_i-\bar{x})(y_i-\bar{y})
= \frac{1}{n-1}\left(\sum_{i=1}^n x_i y_i - n\bar{x}\bar{y}\right)\\
&= \frac{1}{n-1}\left(N_{11} - \frac{N_{1\cdot}\,N_{\cdot 1}}{n}\right)
= \frac{N_{11}N_{00} - N_{10}N_{01}}{n(n-1)},
\end{align*}
where the last equality uses $N_{11} n - N_{1\cdot} N_{\cdot 1} = N_{11}(N_{11}+N_{10}+N_{01}+N_{00}) - (N_{11}+N_{10})(N_{11}+N_{01}) = N_{11}N_{00} - N_{10}N_{01}$. Hence
\[
r = \frac{S_{xy}}{s_x s_y}
= \frac{(N_{11}N_{00} - N_{10}N_{01})/[n(n-1)]}{\sqrt{N_{1\cdot}\,N_{0\cdot}\,N_{\cdot 1}\,N_{\cdot 0}}/[n(n-1)]}
= \frac{N_{11}N_{00} - N_{10}N_{01}}{\sqrt{N_{1\cdot}\,N_{0\cdot}\,N_{\cdot 1}\,N_{\cdot 0}}},
\]
which is the definition of $\phi$. \qedhere
\end{proof}

\begin{remark}[Population analogue]
Let $X, Y \in \{0,1\}$ with $p_x = \Pr(X=1) \in (0,1)$, $p_y = \Pr(Y=1) \in (0,1)$, and joint probabilities $\{p_{ij}\}_{i,j \in \{0,1\}}$. Then $\operatorname{Var}(X) = p_x(1-p_x)$, $\operatorname{Var}(Y) = p_y(1-p_y)$, and $\operatorname{Cov}(X,Y) = p_{11} - p_x p_y = p_{11}p_{00} - p_{10}p_{01}$, so
\[
\rho = \frac{p_{11}p_{00} - p_{10}p_{01}}{\sqrt{p_x(1-p_x)\,p_y(1-p_y)}} = \phi.
\]
\end{remark}

\begin{remark}[Invariance to coding]
Both $r$ and $\phi$ are invariant under any positive affine recoding $x \mapsto \alpha x + \beta$, $y \mapsto \gamma y + \delta$ with $\alpha, \gamma > 0$, provided the convention that the larger of the two coded values labels the "$1$" row (or column) of the contingency table is preserved. Swapping the row labels (or column labels) flips the sign of $\phi$, matching the corresponding sign flip of $r$.
\end{remark}

\begin{remark}[Degeneracy]
If any marginal count is zero or equals $n$ (i.e., $N_{i\cdot} \in \{0, n\}$ or $N_{\cdot j} \in \{0, n\}$), the corresponding sample variance vanishes and both $r$ and $\phi$ are undefined. Assumption~\ref{ass:phi}(iii) rules this out.
\end{remark}

\end{document}